\documentclass[a4paper,reqno, 11pt]{amsart}
\pdfoutput=1
\usepackage{graphicx,color}
\usepackage{footmisc}
\usepackage{amsmath}
\usepackage{amsfonts}
\usepackage{amssymb}
\usepackage{bbm}
\usepackage{epstopdf}
\usepackage{color}
\numberwithin{equation}{section}
\usepackage{tikz}
\usepackage{pgfplots}
\usepackage{subfig}
\usetikzlibrary{fit}
\usetikzlibrary{shapes.geometric}
\usetikzlibrary{decorations.pathmorphing}

\let\a=\alpha   \let\g=\gamma  \let\d=\delta \let\e=\varepsilon
     \let\th=\theta \let\k=\kappa \let\l=\lambda
\let\m=\mu    \let\n=\nu         \let\p=\pi    \let\r=\rho
\let\s=\sigma \let\t=\tau    \let\c=\chi
   \let\o=\omega
\let\G=\Gamma \let\D=\Delta  \let\L=\Lambda

\newcommand{\EE}{{\mathcal E}}

\newcommand{\VV}{{\mathcal V}}
\newcommand{\II}{{\mathcal I}}

\newcommand{\NN}{{\mathcal N}}

\newcommand{\WW}{{\mathcal W}}

\newtheorem{Theorem}{Theorem}

\newtheorem{Remark}{Remark}

\newtheorem{Lemma}{Lemma}

\newtheorem{Proposition}{Proposition}

\def\nn{\nonumber}

\def\\{\hfill\break}
\def\={:=}
\let\io=\infty
\let\0=\noindent

\def\media#1{{\langle#1\rangle}}

\let\dpr=\partial

\def\tende#1{\,\vtop{\ialign{##\crcr\rightarrowfill\crcr\noalign{\kern-1pt
    \nointerlineskip} \hskip3.pt${\scriptstyle #1}$\hskip3.pt\crcr}}\,}
\def\otto{\,{\kern-1.truept\leftarrow\kern-5.truept\to\kern-1.truept}\,}

\def\to{\rightarrow}

\def\qed{\hfill\raise1pt\hbox{\vrule height5pt width5pt depth0pt}}

\def\lis{\overline}

\def\be{\begin{equation}}
\def\ee{\end{equation}}
\def\bea{\begin{eqnarray}}
\def\eea{\end{eqnarray}}
\def\nn{\nonumber}
\def\pref#1{(\ref{#1})}

\def\lb{\label}

\def\openone{{\mathbbm{1}}}

\def\bml{\begin{multline}}
\def\eml{\end{multline}}

\usetikzlibrary{fit}
\usetikzlibrary{shapes.geometric}
\usetikzlibrary{decorations.pathmorphing}
\usetikzlibrary{decorations.markings}
\usetikzlibrary{arrows}

\tikzset{
point/.style={circle,fill=black,inner sep=1pt},
vertex/.style={circle,fill=black,inner sep=1.5pt},   
bvertex/.style={circle,fill=black,inner sep=2.8pt},
Bvertex/.style={circle,fill=black,inner sep=4pt}, 
specialEP/.style={rectangle,fill=white,draw,inner sep=3pt},  
whitevex/.style={circle,fill=white,draw, inner sep=2pt},
linelabel/.style={sloped,above,very near start, inner sep=1pt,execute at begin node=$\scriptstyle,execute at end node=$},
baseline=(current  bounding  box.center),doubled/.style={double distance= 1pt,line width=1.5pt},
th/.style={line width=0.5 pt, gray},  
med/.style={line width=1 pt}  
}

\if \tikzset{vertex/.style={circle,fill=black,inner sep=2pt},
bigvertex/.style={circle,fill=black,inner sep=4pt},
specialEP/.style={rectangle,fill=white,draw,inner sep=3pt},
nuEP/.style={circle,fill=white,draw, inner sep=2pt},
linelabel/.style={sloped,above,very near start, inner sep=1pt,execute at begin node=$\scriptstyle,execute at end node=$},
baseline=(current  bounding  box.center),doubled/.style={double distance= 1pt,line width=1.5pt}
med/.style={line width=1 pt}  
}\fi
\pgfdeclarelayer{background}
\pgfsetlayers{background,main}

\begin{document}
\title{Haldane relation for interacting dimers}
\author{Alessandro Giuliani}
\address{Dipartimento di Matematica e Fisica Universit\`a di Roma Tre \\ \small{
L.go S. L. Murialdo 1, 00146 Roma, Italy}}
\email{giuliani@mat.uniroma3.it}
\author{Vieri Mastropietro}
\address{Dipartimento di Matematica, Universit\`a di Milano \\
  \small{Via Saldini, 50, I-20133 Milano, ITALY }}
\email{vieri.mastropietro@unimi.it}
\author{Fabio Lucio Toninelli}
\address{Universit\'e de Lyon, CNRS, Institut Camille Jordan, Universit\'e Claude Bernard Lyon 1\\
\small{43 bd du 11 novembre 1918,
69622 Villeurbanne cedex, France}}
\email{toninelli@math.univ-lyon1.fr}

\begin{abstract} 
We consider a model of weakly interacting, close-packed, dimers on the two-dimensional square lattice. In a previous paper, we 
computed both the multipoint dimer correlations, which display non-trivial critical exponents, continuously varying with the interaction strength; 
and the height fluctuations, which, after proper coarse graining and rescaling, converge to the massless Gaussian field with a suitable 
interaction-dependent pre-factor (`amplitude'). 
In this paper, we prove the identity between the critical exponent of the two-point dimer correlation and the amplitude of this massless Gaussian field.
This identity is the restatement, in the context of interacting dimers, of one of the 
Haldane universality relations, part of his Luttinger liquid conjecture, originally formulated in the context of one-dimensional interacting Fermi systems. Its validity is a strong confirmation of the effective
massless Gaussian field description of the interacting dimer model, which was guessed on the basis of formal bosonization arguments. 
We also conjecture that a certain discrete curve defined at the lattice level
via the Temperley bijection converges in the scaling limit to an SLE$_\k$
process, with $\k$ depending non-trivially on the interaction and related
in a simple way to the amplitude of the limiting Gaussian field.
\end{abstract}

\maketitle

\section{Introduction}

In recent years, there has been an increasing interest of the mathematical community on the conformal invariance properties 
of two-dimensional (2D) statistical systems at the critical point, and on their connections with the massless Gaussian field. The introduction of 
novel techniques, ranging from discrete holomorphicity \cite{Ca,Ke1,Ke2,Sm10} to Schramm-Loewner Evolution (SLE) \cite{Law} and percolation techniques \cite{We}, finally allowed, after 
more than 50 years of intense research, to fully characterize the structure and the conformal invariance of dimer \cite{Du_dimers,Ke} and Ising models \cite{CHI,HS}, as well as to rigorously 
confirm some predictions, based on Conformal Field Theory (CFT) arguments, concerning crossing probabilities in critical percolation \cite{Sm_p}. 
This exciting advances are still ongoing and, as they develop, 
they are revealing a closer and closer connection between the lattice integrability properties of dimer and Ising models, 
and certain CFT structures and objects, such as the Virasoro algebra, the Operator Product Expansion, and the stress-energy tensor \cite{HVK}.

A big limitation of these methods is that they are mostly limited to models at the free fermion point, and it is a major challenge to develop techniques for rigorously controlling 
the scaling limit of general interacting, non-integrable, theories, and for proving their conformal covariance properties.
A standard method used in the Quantum Field Theory (QFT) and condensed matter communities for characterizing {\it quantitatively} the scaling limit 
of 2D interacting theories at the critical point is the so-called `Coulomb gas formalism', which is based on an effective description of several 
2D critical theories in terms of a massless Gaussian field (`bosonization method') \cite{LP,Nei}. In this approach, the conformal invariance of the interacting theory translates into the well-known 
conformal invariance of the Gaussian model. Most of the applications of this method are quite heuristic, and we are still missing a full comprehension of the emergence of the 
massless Gaussian field in the scaling limit. 

In this paper, we prove a rigorous instance of this emergent correspondence in the context of interacting dimer models, 
giving a strong justification of the use of the Coulomb gas description outside the free fermion point. Our new result complements our previous findings in \cite{GMTlungo},
where we showed that the critical exponent of the dimer-dimer correlations is an analytic function of the interaction strength, 
and proved the convergence, at large distances, of the height function to the massless Gaussian field,
with a suitable interaction-dependent pre-factor. In short, our new result is a rigorous proof of the identity between such pre-factor (the `stiffness' or `amplitude' of the massless Gaussian field)
and the dimer-dimer critical exponent. This identity, very surprising at first sight, can be guessed on the basis of the aforementioned representation of the interacting dimer model 
in terms of a massless Gaussian field. It is an instance of the so-called Haldane relations, originally formulated in the context of one-dimensional interacting Fermi systems, 
as parts of the famous Haldane's Luttinger liquid picture \cite{H,Ha2}. 

Before formulating the model and the result precisely (see Section \ref{sec:model}), we first 
make a historical digression on the concept of Luttinger liquid universality class and on the bosonization method, 
which may be useful for clarifying the motivations behind the emergence of an effective `massless Gaussian  Field'  behavior and the Haldane relation. 
We also explain the connection between these concepts and 2D dimer models.

\subsection{Luttinger liquids.} 
Kadanoff \cite{K77} and Haldane \cite{H,Ha2} (see also \cite{dN81,KB79,KW71,LP,PB81}) {\it conjectured} the existence of a {\it universality class}, called {\it 8-vertex universality class} or {\it Luttinger liquids}, of models whose low energy behavior is described by a 2D massless
Gaussian field (i.e., a free massless boson field) $\phi$ with covariance 
\begin{equation}
\label{cov:GFF}
\mathbb E(\phi(x)\phi(y))=-\frac{A}{2\pi^2}\log|x-y|,
\end{equation}
for a suitable constant $A$, which is model dependent. In particular,
the correlation functions of models in this class are the same, asymptotically at large distances, as those of suitable functions
of $\phi$. Models in this class include two-dimensional classical systems, like:  
the 6 and 8-vertex models, the Ashkin-Teller model, and interacting dimer models at close-packing; 
and one-dimensional quantum systems, like: 
Heisenberg spin chains, the Luttinger model and the spinless Hubbard models (see  \cite{Ba_book,H,K77,LP,PB81}). This conjecture implies, 
in particular, that the critical exponents are connected by {\it scaling relations} such that, once a single exponent
is known, all the others are determined, and that there are simple
relations between the critical exponents and the {\it amplitudes} (i.e., the pre-factors in front of the power-law decay)
of suitable correlations. These predictions have been checked mostly in exactly solvable models, but they 
are expected to hold also for non solvable ones.

\subsection{Bosonization.}
The models in the Luttinger liquid universality class that we mentioned describe interacting spins or
particles: therefore, the mapping of all these systems into such a simple model as the massless Gaussian field
is at first sight surprising. The simplest way to understand this correspondence is via 
the concept of {\it bosonization}, which is a
crucial notion in 2D QFT and in condensed matter physics. The starting point is the
observation that all the lattice models in the Luttinger liquid universality class (vertex models, spin chains, dimers, Ashkin-Teller)
admit an exact description in terms of lattice fermions, i.e., a family of Grassmann
variables $\psi_{x,\o}^\sigma$, indexed by lattice vertices $x=(x_1,x_2)$, as well as by two 
indices $\sigma,\omega=\pm$. For instance, the $8$-vertex and the Ashkin-Teller models can be represented as a pair of 2D Ising 
models coupled via a quartic interaction \cite{Ba_book}, and the fermionic 
representation is inherited from the Pfaffian description of the 2D Ising
model. For special values of the model parameters (free-fermion
point) such fermions are non-interacting and then the system is exactly solvable. However, for
generic values of the parameters the fermions are interacting, i.e., their
action contains terms at least quartic in the Grassmann variables, so that the partition function and the correlations are 
given by non-Gaussian Grassmann integrals.
If one performs a formal continuum limit, such fermions becomes
Dirac fermions in $d=1+1$ dimensions, which are massless at
criticality.  

There is a well known correspondence in $d=1+1$ Quantum Field Theory
between bosons and fermions \cite{C}. Take non-interacting
massless Dirac fermions $\psi^\sigma_{x,\o}$, $\sigma,\o=\pm$, $x\in\mathbb R^2$, with
propagator that is diagonal in the index $\omega$, anti-diagonal in $\sigma$, translation-invariant in $x$ and such that 
\be
\label{propintro}
\langle \psi^{-}_{x,\omega}\psi^{+}_{0,\omega}\rangle= \frac{C_\o}{x_1+ i\omega x_2},\ee 
with $C_\omega$ constants such
that\footnote{A standard choice is $C_\omega=1/(2\pi)$. However, for the interacting dimer model 
we consider below a natural choice of coordinates leads to 
   $C_\omega$ depending on $\o$. We could reduce to the standard
  case via a suitable rotation of space coordinates, but then many formulas would look more
  cumbersome.} $C_-=C_+^*$. Then bosonization consists in the following two
identities (see e.g. \cite{D}):
\begin{itemize}
\item the ``fermionic
density'' $\psi^+_{x,\o} \psi^-_{x,\o}$ has the same multi-point correlations as the derivative of a boson field:
\be \psi^+_{x,\o} \psi^-_{x,\o}\quad \otto \quad 2i \pi C_\o \partial_\o\phi\,,\label{dphi}\ee where $\partial_\o:=\frac12(\partial_{x_1}- i\o \partial_{x_2})$, and 
$\phi$ is the massless Gaussian field with covariance \eqref{cov:GFF} with
$A=1$. In particular, correlations of $\psi^+_{x,\o} \psi^-_{x,\o}$ of odd order  and truncated correlations of order larger than $2$ vanish.
\item the ``fermionic mass'' $\psi^+_{x,\o}
\psi^-_{x,-\o}$ has the same correlations as a normal-order exponential of the boson field: 
\be \psi^+_{x,\o}
  \psi^-_{x,-\o} \quad \otto \quad C_\o:e^{2\pi i\omega \phi(x)}:\,,\label{eiphi}\ee where $:\dots:$ denotes Wick ordering, see
e.g. \cite{D}.
\end{itemize}

Remarkably, a similar correspondence remains valid \cite{C} also for {\it interacting}
massless Dirac fermions. In particular, in the presence of a local, 
delta, interaction (Thirring model) the multi-point correlations of the fermionic density and fermionic mass operators 
are known explicitly \cite{J,Kl,Th} and the bosonization identities are still true, provided  
that the left side of \eqref{eiphi} is divided by a suitable 
renormalization constant, diverging in the ultraviolet limit. Moreover, the pre-factor $A$ 
in \eqref{cov:GFF} is changed into an interaction-dependent constant, $A=A(\l)\neq 1$, where $\l$ is the strength of the 
delta interaction. Note that the density and mass operators in the Thirring model are naturally defined up to 
multiplicative renormalization constants, whose specific choices  
are part of the definition of the model. In this sense, the `amplitudes' of the 
fermionic operators (which play a role in the Haldane relations we are interested in)
are somewhat ambiguous: therefore, it would be desirable to have at disposal a 
solvable model, similar to Thirring, but free of  ultraviolet divergences, to be used to test unambiguously the 
desired universality relations between amplitudes and critical exponents. 

The simplest such model is a model of interacting massless Dirac fermions, 
whose interaction is local, delta-like, only in one of the two directions, say horizontal (Luttinger model): 
also in this case the correlations can be computed exactly \cite{ML_lut} and, 
for a proper choice of the bare Fermi velocity (chosen in such a way that the interacting Fermi velocity is 1),
their asymptotic expression  at large distances reads: 
\be\sum_\o \langle \psi^{+}_{x,\omega}\psi^{-}_{x,\omega};
\psi^{+}_{0,\omega}\psi^{-}_{0,\omega}\rangle\simeq \frac{A}{2\pi^2} \frac{x_1^2-x_2^2}{|x|^4}\,,\label{eq:1.5}\ee
and 
\be \sum_\o \langle \psi^{+}_{x,\omega}\psi^{-}_{x,-\omega};
\psi^{+}_{0,-\omega}\psi^{-}_{0,\omega}\rangle\simeq \frac{B}{2\pi^2} \frac{1}{|x|^{2A}}\,\label{eq:1.6}\ee
where the semicolon indicates truncated expectation, and $A$ and $B$ are suitable interaction-dependent constants. 
Eq.\eqref{eq:1.5}-\eqref{eq:1.6} are the same that one would obtain by using \eqref{dphi}-\eqref{eiphi}, with $\phi$ normalized as 
in \eqref{cov:GFF} and $C_\o=1/(2\pi)$, provided the right side of \eqref{eiphi} is multiplied by $\sqrt{B}$. 
Note that the constant $A$ in \eqref{eq:1.5} is the same as the one appearing in the critical exponent in \eqref{eq:1.6}. The same identity
between the amplitude of the density-density correlations and the critical exponent of the mass-mass correlations
has been verified for other exactly solvable models in the same universality class, in particular for the XXZ spin chain \cite{H}. The conjecture is that 
it should remain valid also for non-exactly solvable models in the same universality class, including lattice models, for which 
an exact mapping into a massless Gaussian field is not possible. 

\subsection{Previous results.} In the absence of exact solutions or of bosonization identities, the computation of the asymptotic behavior of correlations 
at criticality, not to mention the verification of the Kadanoff-Haldane relations, is a major mathematical challenge. 
Constructive QFT and Renormalization Group (RG) techniques provide powerful mathematical tools 
to attack these problems. These methods allowed one of us \cite{M} to prove that the critical exponents of the 8-vertex and the Ashkin-Teller 
model, close to the free fermion point, are analytic functions of the interaction strength $\l$. Benfatto, Falco and Mastropietro \cite{BFM} 
later extended the analysis of \cite{M}, proving the validity of some of the Kadanoff relations between the critical exponents 
of several non-integrable models in the 8-vertex universality class, as well as two Haldane relations 
between exponents and amplitudes for a quantum spin chain \cite{BMun,BMdr}.
These results provide relations only between correlations of quantities that are
\emph{local} in the fermionic variables. An important open problem is to analyze observables that are, instead,
\emph{non-local}, such as the spin in coupled Ising layers and the height function in dimer models.  The analysis of these 
observables is already very non-trivial at the free-fermion point \cite{CHI,Du_dimers,Du_bo,Ke2,Ke}.

\subsection{Interacting dimers.} 
Motivated by these issues, in \cite{GMTlungo,GMTcorto} we considered a non-integrable dimer model at close packing on $\mathbb Z^2$, 
where dimers interact via a short-range potential of strength $\lambda$. We studied two natural observables: 
the dimer occupation variable $\openone_e$, indexed by edges $e$ of the square lattice, 
and the height function $h(\eta)$, indexed by faces $\eta$.
The model can be rewritten (see Section \ref{sec:intmod}
below) as a system of two-dimensional interacting lattice Dirac fermions with
propagator behaving at large distances as in \eqref{propintro} with
$C_\o=1/(2\pi(1-i \o))$ (see \eqref{eq:16bis}). When the coupling
parameter $\lambda$ is set to zero the fermion-fermion interaction
vanishes, corresponding to the Pfaffian nature of Kasteleyn's solution
of the non-interacting dimer model. The observable $\openone_e$
has a local expression $\mathcal I_e$ in the fermionic
representation, containing both a ``density term'' $\psi^+_{x,\o}\psi^-_{x,\o} $ and
a ``mass term'' $\psi^+_{x,\o}\psi^-_{x,-\o}$, the latter multiplied
by a prefactor that oscillates with distance. If, e.g., $e$ is horizontal, then  
(if $x=(x_1,x_2)$ is the coordinate of the edge $e$ and $(-1)^x=(-1)^{x_1+x_2}$)
\begin{eqnarray}
\label{peciona1}
\mathcal I_e=\pm\sum_\o[\psi^+_{x,\o}\psi^-_{x,\o}+(-1)^x\psi^+_{x,\o}\psi^-_{x,-\o}]+h.o.,
\end{eqnarray}
where the sign in front depends on the parity of the edge,
and $h.o.$ indicates subdominant terms\footnote{By `subdominant' here we mean contributions that produce faster decaying 
terms at large distances in the computation of $\media{\mathcal I_e; A_{e'}}$, where $A_{e'}$ is an observable localized around the bond $e'$, and $|e-e'|\gg 1$.}; 
see below for the full expression (see in particular Remark \ref{rmk:4}).
The height function $h(\eta)$ is defined as the sum over a lattice path 
$\eta_0,\eta_1,\dots,\eta_k=\eta$, from a reference face
$\eta_0$ to $\eta$, of $\s_e(\openone_e-1/4)$, with $\s_e$ a suitable sign (see \eqref{altezza} and following lines), and, therefore, by \eqref{peciona1}, 
it is non-local in the fermionic variables.

%


Using \eqref{peciona1}, one finds that the dimer-dimer correlation $\langle
\openone_e;\openone_{e'}\rangle_\lambda$ is the sum of two terms, 
the first corresponding to a density-density correlation, 
\be
\sum_\o \langle  \psi^+_{x,\o}\psi^-_{x,\o}; \psi^+_{x',\o}\psi^-_{x',\o}\rangle\label{eq:1.8}
\ee
and the second to a mass-mass correlation with oscillating prefactor:
\be
\label{eq:oscill1}
(-1)^{x+x'}\sum_\o\langle  \psi^+_{x,\o}\psi^-_{x,-\o};
\psi^+_{x',-\o}\psi^-_{x',\o}\rangle.\ee
By heuristically using the bosonization identities \eqref{dphi}-\eqref{eiphi}, 
one guesses that \eqref{eq:1.8} should behave asymptotically as
\be
\label{eq:boso1}
\frac{A}{2\p^2}{\rm Re}\left(\frac{1}{(1-i)^2(z-z')^2}\right),
\ee
with $z=x_1+i x_2,z'=x_1'+i x_2'$,
while \eqref{eq:oscill1} should be asymptotically proportional to 
\begin{eqnarray}
\label{eq:boso2}
(-1)^{x+x'}\frac1{|z-z'|^{2\n}},
\end{eqnarray}
with $\n=A$. Similarly, by using the fermionic representation of the height function and the first bosonization identity 
\eqref{dphi}, one predicts that $h(\eta)$ tends to a massless Gaussian field with covariance \eqref{cov:GFF}, and the same $A$ as in 
\eqref{eq:boso1}.

In \cite{GMTlungo,GMTcorto}, we developed a Renormalization Group
analysis for the interacting dimer model, which allowed us to prove
most of these predictions, in particular the convergence of the height
field to the massless Gaussian field, the validity of \eqref{eq:boso1}
and \eqref{eq:boso2} and their multipoint generalization.  However,
the important question of the identity $A=\nu$ was not addressed, and
we prove it here, by combining the ideas of \cite{GMTlungo,GMTcorto},
together with a comparison of exact lattice Ward identities with those
of a relativistic reference model, in the spirit of \cite{BMun,BMdr}. The
remarkable identity $A=\nu$ is a strong confirmation of the validity of the
massless Gaussian field description of the interacting dimer model. 
It is a restatement of the Haldane relation between the `compressibility' and the `density-density critical index', 
in a context different from the one originally proposed by Haldane, who 
considered interacting fermions in one dimension and quantum spin chains: in this sense, it is 
the first example of such a universality relation in a classical statistical mechanics model. 

\subsection{Summary} 

The rest of the paper is organized as follows. 
In Section \ref{sec:model} we define the model precisely and state our main results. We also comment about possible generalizations to related models, and 
about the possible emergence, in the scaling limit, of an SLE$_\k$ process with $\l$-dependent diffusion constant $\k$, underlying the interacting dimer model. The following sections contain the technical 
aspects of the proof of our main theorem, in particular: \\
-- in Section \ref{section3} we discuss the Grassmannian formulation of the model; \\
-- in Section \ref{sec1} we derive a lattice Ward Identity for the dimer model, which plays an important role in the 
proof of the Haldane relation; \\
-- in Section \ref{sec:relativistico} we introduce a relativistic continuum model (the `reference model'), which plays the role of 
the infrared fixed point of the dimer model, and review some of the properties of its correlations; \\
-- in Section \ref{sec8} we put together the ingredients of the previous sections and complete the proof of the Haldane relation;
in Section \ref{sec9} we compare the notations and conventions of the present paper with those of \cite{GMTlungo}.

In the appendices we collect a few more technical issues: in Appendix \ref{app:residui} we discuss the structure of the singularities 
of the Fourier transform of the dimer-dimer correlations; in Appendix \ref{app:primordine} we verify the Haldane relation at first order in perturbation theory;
in Appendix \ref{app3} we review the derivation of the Ward Identities for the relativistic reference model, as well as the exact computation of its two- and four-point functions. 

\section{Model and result}
\label{sec:model}

Let us now define our model more precisely. 
We consider the bipartite graph $\mathbb Z^2$, and we decompose it into two sublattices (black/white sublattices $V_B/V_W$) such that 
all neighbors of a vertex $v\in V_B$ belong to $V_W$ and viceversa.
Each vertex is assigned a coordinate $x=(x_1,x_2)$ and 
a white vertex has the same coordinates as the black
vertex just at its left. Both $V_B$ and $V_W$ will be thought of as Bravais lattices with basis vectors $\vec e_1$ and $\vec e_2$, 
where $\vec e_1$ is the vector of length
$\sqrt 2$ and angle $-\pi/4$ w.r.t. the horizontal axis, while $\vec e_2$ is the one of length
$\sqrt 2$ and angle $+\pi/4$. See Figure \ref{fig:sub1}.

\begin{figure}
\centering
  \includegraphics[width=.5\linewidth]{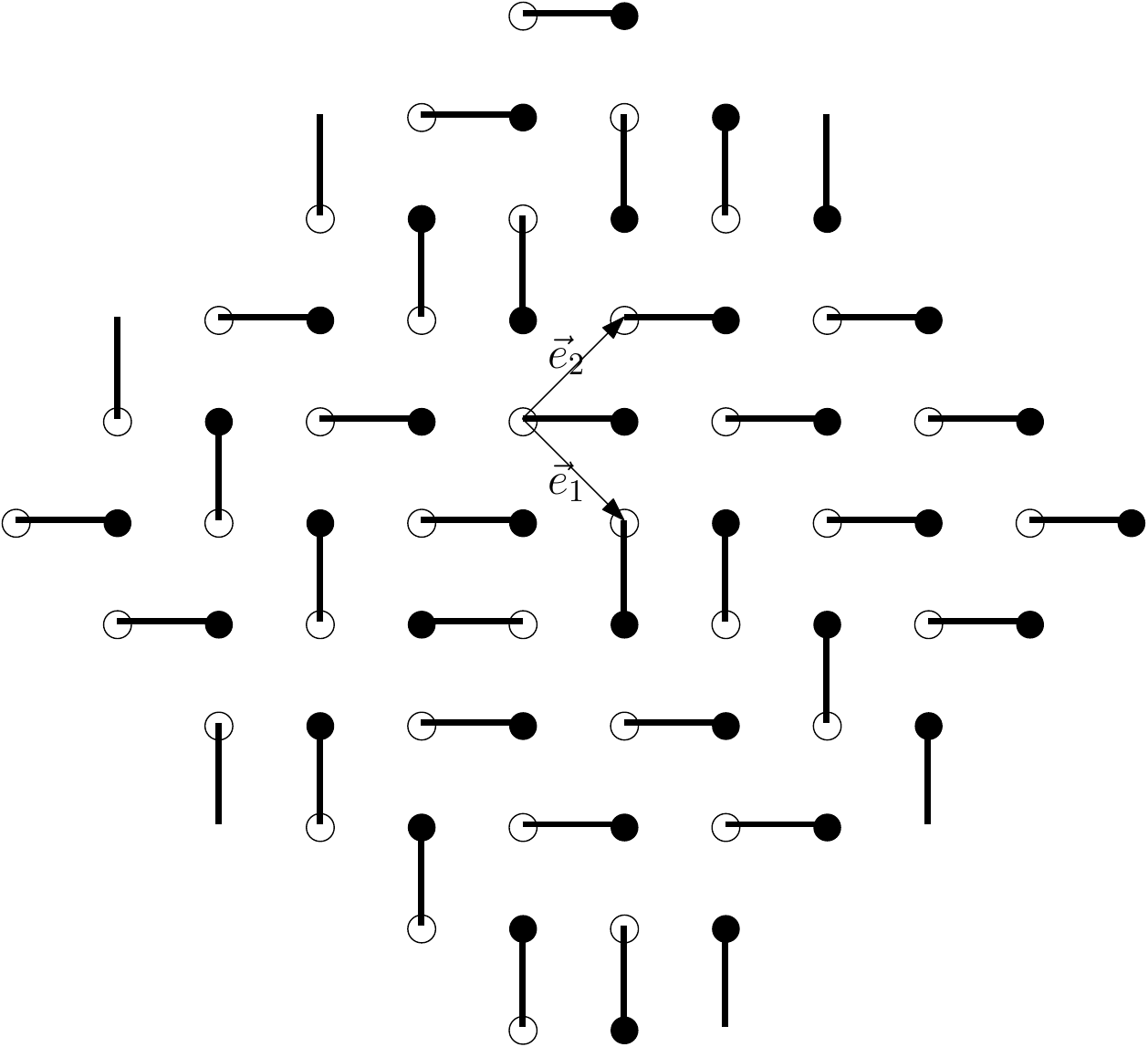}
  \caption{The sublattices $V_B$ and $V_W$ of black and white sites, and the basis vectors $\vec e_1, \vec e_2$, 
  oriented at an angle $\mp \p/4$ with respect to the horizontal axis. The figure also shows an admissible close-packed dimer configuration, 
 periodic of period $L=6$ in the directions $\vec e_1,\vec e_2$.}
  \label{fig:sub1}
\end{figure}

Given an edge $e$, we let $b(e),w(e)$ denote the black/white vertex of
$e$.  Edges are of four different types $r=1,2,3,4$: we say that $e$
is of type $r$ if the vector from $b(e)$ to $w(e)$ forms an
anti-clockwise angle $(r-1) \pi/2$ with respect to the horizontal
axis. An edge $e$ is unambiguously identified by its type $r(e)$ and by the coordinates $x(e)=(x_1(e),x_2(e))$ of its black site.

We consider the dimer model in a periodic box: we let $\mathbb T_L$ be
the graph $\mathbb Z^2$ periodized (with period $L$) in both directions $\vec
e_1,\vec e_2$. See Figure \ref{fig:sub1}. With abuse of notation, we still denote by $V_W,V_B$ the
set of white/black sites of $\mathbb T_L$, without making the $L$
dependence explicit. Black/white sites are therefore indexed by
coordinates $x\in \L=\L_L=\{(x_1,x_2), 1\le x_i\le L\}$.

The partition function  of the interacting dimer model we study is
\begin{eqnarray}
  \label{eq:Z}
  Z_{L}(\lambda,m)=\sum_{M\in\mathcal M_L} (\prod_{e\in M} t_e^{(m)})e^{\lambda H_L(M)}=:\sum_{M\in\mathcal M_L}p_{\Lambda;\lambda,m}(M)
\end{eqnarray}
where:
\begin{itemize}
\item $\mathcal M_L$ is the set of perfect matchings of $\mathbb T_L$;
\item $H_L(M)$ is the number of square plaquettes of $\mathbb T_L$ containing two
parallel dimers;
\item $\lambda$ is a real parameter (coupling constant);
\item $m>0$ and
\begin{eqnarray}
  \label{eq:tem}
  t_e^{(m)}=1+m(-1)^{x_1(e)+x_2(e)}\left(\delta_{r(e)=1}-\delta_{r(e)=3}\right).
\end{eqnarray}

\end{itemize}
As discussed in \cite{GMTlungo}, the parameter $m\ge0$ (the mass) plays the role of an infrared cut-off, to be eventually sent to zero: it has the effect that
correlations decay exponentially with distance (uniformly in $L$) as
long as $m>0$, and it is sent to zero after the thermodynamic limit
$L\to\infty$.  This model, in the limit $m\to 0$, describes polar crystals \cite{HP} and it was recently
reconsidered in \cite{Alet,PLF} in connection with
quantum dimer models.

The Boltzmann-Gibbs measure associated with the model is denoted by $\langle\cdot\rangle_{\Lambda;\lambda,m}$: if $O(M)$ is a function of the dimer configuration, 
\begin{eqnarray}
  \label{eq:Om}
 \langle O\rangle_{\Lambda;\lambda,m}
  :=\frac1{Z_\Lambda(\lambda,m)}\sum_{M\in\mathcal M_L} p_{\Lambda;\lambda,m}(M)O(M).
\end{eqnarray}

\begin{Remark}
\label{rem:mc1}
This model is the same studied in \cite{GMTlungo,GMTcorto}. However, here we use a different
convention for the coordinates on the lattice, which respects the bipartite structure and turns out to be convenient for the derivation of the Ward Identities. 
In order to restate the results derived here in the notations of \cite{GMTlungo}, one needs to properly redefine the coordinates: with the conventions of 
\cite{GMTlungo,GMTcorto}, the black (resp. white) site of coordinate $x=(x_1,x_2)$ have coordinates $\tilde x(x)$ (resp. $\tilde x(x)+(1,0)$), where
\begin{eqnarray}
  \label{eq:tildex}
\tilde x(x)=(x_1+x_2,x_2-x_1).
\end{eqnarray}
\end{Remark}

As discussed in \cite{GMTlungo}, whenever $O$ is a bounded local function, the following limit exists:
\begin{eqnarray}
  \label{eq:limite}
  \langle O\rangle_{\lambda}:=\lim_{m\to0}\lim_{L\to\infty} \langle O\rangle_{\Lambda;\lambda,m} ,
\end{eqnarray}
and defines a translationally invariant infinite volume Gibbs state $\media{\cdot}_\l$. 

When $\lambda=0$ (non-interacting model) the model is well-known to be
exactly solvable via Kasteley's theory \cite{K1}, in the sense that $n$-point
correlations can be computed explicitly as determinants.
When both $\lambda$ and $m$ are zero
(non-interacting and massless model) the partition function reduces to the
cardinality of $\mathcal M_L$.


Before stating our main result, let us review briefly what was proven in \cite{GMTlungo}. 
Recall that the height function is defined by fixing it to zero at some reference face $\eta_0$, and by establishing that 
\begin{eqnarray}
\label{altezza}
h(\eta)-h(\xi)=\sum_{e\in C_{\xi\to\eta}}\sigma_e (\openone_e-1/4),  
\end{eqnarray}
with $C_{\xi\to\eta}$ any nearest-neighbor path from $\xi$ to $\eta$, the sum running over the edges crossed by the path,  $\openone_e$  the indicator function
that an edge of $\mathbb Z^2$ is occupied by a dimer and $\sigma_e$ being $+1$ or $-1$ according to whether the edge is traversed with the 
white vertex on the right or left.
The content of \cite[Th. 1 and 3]{GMTlungo} is that there exists $\lambda_0>0$ such that, if $|\lambda|<\lambda_0$ then the height field associated to the 
dimer configuration converges in distribution (in the limit $\lim_{m\to0}\lim_{L\to\infty}$) to a massless Gaussian field $\phi(\cdot)$ on the plane, with covariance 
\begin{eqnarray}
  \label{eq:varia}
\mathbb E(\phi(x) \phi(y))=-\frac{A(\lambda)}{2\pi^2}\log|x-y|,
\end{eqnarray}
where\footnote{The constant $A(\lambda)$ was called $K(\lambda)$ in \cite{GMTlungo}. Here we change notation, in order to avoid
confusion with the elements of the Kasteleyn matrix, which is traditionally denoted $K$.}  $A(\lambda)$ is an analytic function of $\lambda$ satisfying $A(0)=1$. Moreover the $n$-th cumulant of $h(\eta)-h(\xi)$, $n\ge 3$, is bounded uniformly in $\eta,\xi$. 

A crucial ingredient in the proof was a sharp asymptotic expression for
dimer-dimer correlations.   The rewriting 
of \cite[Th. 2]{GMTlungo} with the present convention is the following: There exist real analytic functions
  $B(\cdot),\nu(\cdot)$, defined in a neighborhood
  $|\lambda|\le \lambda_0$ of the origin, satisfying
  $B(0)=\nu(0)=1$, such that the following holds.  Let
  $e,e'$ be edges of type $r,r'$, with $b(e)=x=(x_1,x_2)\ne0$ and
  $b(e')=0$. Then, for $|\lambda|\le \lambda_0$,
  \begin{gather}
    \label{eq:28}
    \langle
    \openone_e;\openone_{e'}\rangle_\lambda=-\frac{A(\lambda)}{2\pi^2}{\rm Re}\left[
    \frac{e^{i\frac{\pi}2(r+r')}
    }{(1-i)^2(x_1+ix_2)^2}\right]
\\
+t_{r,r'}\frac{B(\lambda)}{4\pi^2}(-1)^{x_1+x_2}\frac1{|x|^{2\nu(\lambda)}}+R_{r,r'}(x)\nonumber
  \end{gather}
where 
\begin{eqnarray}
  \label{eq:29}
  |R_{r,r'}(x)|\le C_\theta(1+|x|)^{-2-\theta}
\end{eqnarray}
for some $1/2\le \theta<1$, $C_\theta>0$, 
and $t_{r,r'}$ is $1$ if $r=r'$, $0$ if the two edges are not parallel, and $-1$ if they are parallel but not of the same type. 
The function  $A(\cdot)$ is the same as in \eqref{eq:varia}. The first line of \eqref{eq:28} coincides with the right side of \eqref{eq:boso1} for $r=r'=1$ and $z'=0$.

Note that the large-scale behavior of the height field depends only on
$A(\cdot)$ and not on $B(\cdot),\nu(\cdot)$. Observe also that the decay at large distances
of the dimer-dimer correlation is controlled by the critical exponent
$\min(2,2\nu(\lambda))$.

The main result of the present work is an identity between the limit
variance $A(\l)$ of the height field and the dimer-dimer critical exponent $\nu(\l)$. 
\begin{Theorem} 
\label{th:k=k} There exists $\l_0$ such that, if $|\l|\le \l_0$, then \eqref{eq:28} holds with \be \boxed{A(\lambda)=\nu(\lambda)}\label{Hr}\ee
\end{Theorem}

This result confirms the predictions of the universality conjecture of 
Kadanoff and Haldane for this model, as discussed in the Introduction. 
It would be interesting to apply the methods of its proof to the computation of
other universality relations, such as the relation between the sub-leading corrections to the 
free energy of the interacting dimer model and the central charge \cite{A,BCN}, in the spirit of \cite{GMcc}.

\begin{Remark}[First order computation]
While the non-perturbative proof of Theorem \ref{th:k=k} (Sections
\ref{sec1} to \ref{sec9}) requires the use of lattice Ward
identities and a comparison with a continuum model, one can check directly 
\eqref{Hr} at low orders in perturbation theory. 
In Appendix \ref{app:primordine} we check the equality at first order in $\lambda$, via an
explicit computation of the lowest-order Feynman diagrams. Remarkably, even at lowest order,
this equality requires non-trivial cancellations between Feynman diagrams.
As a byproduct of Appendix \ref{app:primordine}, we find that 
\begin{eqnarray}
  \label{eq:Kkl}
  A(\lambda)=\nu(\lambda)=1-\frac4\pi\lambda +O(\lambda^2).
\end{eqnarray}
In view of \eqref{cov:GFF}, this shows that when the dimer-dimer
interaction $\lambda H_L$ is attractive ($\lambda>0$), the variance
height fluctuations decreases (the interface is more rigid) and
dimer-dimer correlations decay slower. This is compatible with the
fact that at large enough $\lambda$ the model is known to have a
rigid, crystalline, phase, characterized by long-range order of
`columnar' type and $O(1)$ height fluctuations \cite{HP}.
\end{Remark}

\subsection{Outlook and conjectures} \\

\vskip-.3truecm

{\it Extensions to other models.}
As discussed in \cite{GMTlungo}, the specific form of the interaction in \eqref{eq:Z} is unimportant for the validity of 
the massless Gaussian behavior of the height function, see Remark 3 after Theorem 3 in \cite{GMTlungo}. 
Similarly, it is unimportant for the validity of the Haldane relation \eqref{Hr}: the same identity holds true 
for a wide class of interacting dimer models, whose interaction is weak, finite range, and symmetric under the natural lattice 
symmetries of $\mathbb Z^2$ (translations, reflections, discrete rotations). Both the results of \cite{GMTlungo} and those of the present paper 
are presumably valid also for other closely related models, in particular for the 6 vertex (6V) model, which is known to be equivalent to an interacting dimer model
on $\mathbb Z^2$ with a plaquette interaction proportional to the
number of \emph{even} faces of $\mathbb Z^2$ with two parallel dimers, see \cite{Ba_8V,Fa6V}. Note that the mapping of the 6V model into such an interacting dimer model preserves 
the height function, up to a factor $1/2$: the height function of the 6V model (defined as in \cite{vB})
equals half the height function of the dimer model, restricted to faces of odd parity, see the comment after Eq.(7) of \cite{Fa6V}. 
The plaquette interaction of the effective dimer representation of the 6V model, acting only on even faces, is not invariant under the full group of translations; 
therefore, it has a slightly different symmetry than the one of the model considered in the present paper and in \cite{GMTlungo}. 
It is likely that such a change is unimportant for the proofs in \cite{GMTlungo}. The explicit verification 
that such a modified symmetry neither changes the structure of the effective 
infrared theory, nor the structure of the Ward Identities, will be discussed 
elsewhere.

\medskip

{\it Emergent SLE.} The emergent description of the interacting dimer 
model in terms of a massless Gaussian field calls for the emergence of
an SLE process, dual to the Gaussian field: in fact, at the continuum level, 
there are several known connections, or `couplings', between these two types of stochastic processes \cite{CaSLE,Du,MS1,MS2,MS3,MS4,ScSh}.
Inspired by these results, we conjecture that a microscopic
geometric curve, associated with the Temperley spanning  forest, defined below,
converges in the scaling limit to a variant\footnote{Given the speculative nature of the discussion, we are on purpose a bit sloppy on the precise nature of the limiting SLE, as well 
as on the role played by the boundary conditions. The `variant' of SLE we refer to is described in detail in \cite{MS1,MS2,MS3,MS4}.} of the space-filling
SLE$_{\kappa'}$ process, with $\kappa'>4$ the largest root of 
\begin{eqnarray}
\label{eq:Akappa}
A(\l)=\frac{2\k'}{(4-\k')^2}.
\end{eqnarray}
Note that, if $\l=0$ (in which case $A=A(0)=1$), then $\k'=8$.

Here we give some support to this conjecture. The starting point is the Temperley bijection, see for instance \cite{DuG,KPW}. Given a dimer
configuration on $\mathbb T_L$, for every white vertex  $w$ of even
parity draw an oriented edge of length $2$ from $w$, that goes along the unique
dimer with endpoint at $w$ (such edge of course ends at another white
vertex $w'$ of the same parity). The collection of edges thus drawn
forms a cycle-rooted spanning forest (CRSF): every connected component
contains an oriented cycle and the other edges in the same component form
trees oriented toward the cycle, to which they are rooted. This CRSF spans the sub-graph of $\mathbb T_L$ induced by white vertices of even parity. See Fig.\ref{fig:sub2},  where the CRSF, that in this example has a single connected component, is drawn in red. Repeating the same
construction with white edges of odd parity one obtains a ``dual''
CRSF (drawn in blue in Fig.\ref{fig:sub2}). Finally, one can draw an oriented wiggly
curve $\Gamma$ that runs between the primary and dual CRSF (in Fig.\ref{fig:sub2} this curve has two connected components, drawn in  different
colors).  
Note that the curve $\Gamma$ passes once through every face of $\mathbb T_L$. We use the convention that, at the center of each face, $\Gamma$ is tangent to 
one of the two main diagonals of the square lattice. 

\begin{figure}
  \centering
  \includegraphics[width=.5\linewidth]{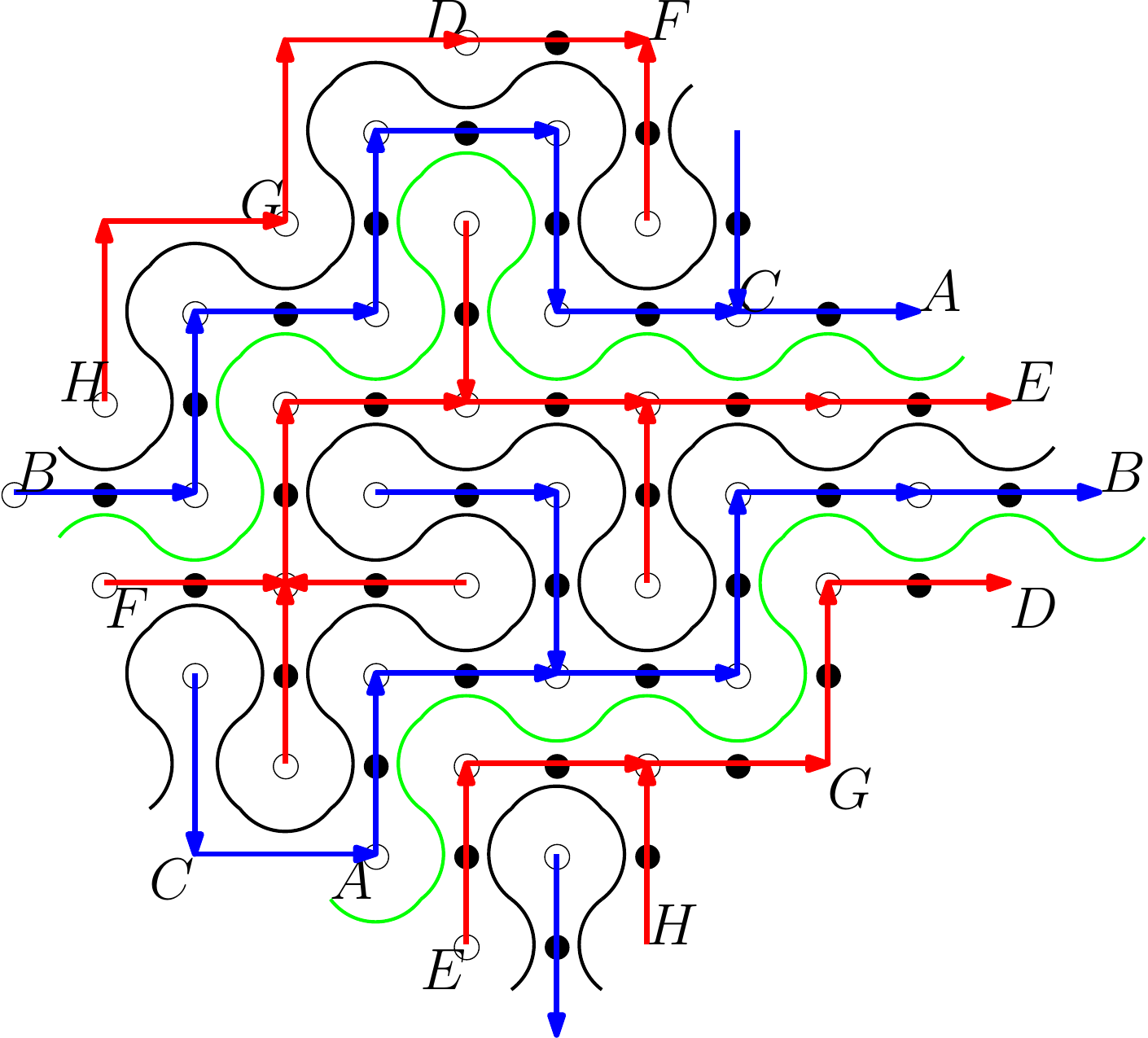}
  \caption{The primary
  and dual CRSF (in red and blue) and the curve $\Gamma$, associated with the dimer configuration of Fig.\ref{fig:sub1}. To help the
  reader follow the curve, we have indicated with the same letter
  points that are identified on the torus.}
  \label{fig:sub2}
\end{figure}

There is a simple correspondence between the height function and the
curve $\Gamma$: as the reader may verify by comparing Figure \ref{fig:sub2} and 
\ref{fig:fig3}, if $\eta_1$ and $\eta_2$ are two faces, which the same connected component of $\G$ passes through, 
the combination $2\p [h(\eta_2)-h(\eta_1)]$ is the net amount of winding of $\G$ between $\eta_1$ and $\eta_2$.
Moreover, such condition on the winding determines the curve $\Gamma$ uniquely.
\begin{figure}
\centering
\includegraphics[width=.5\linewidth]{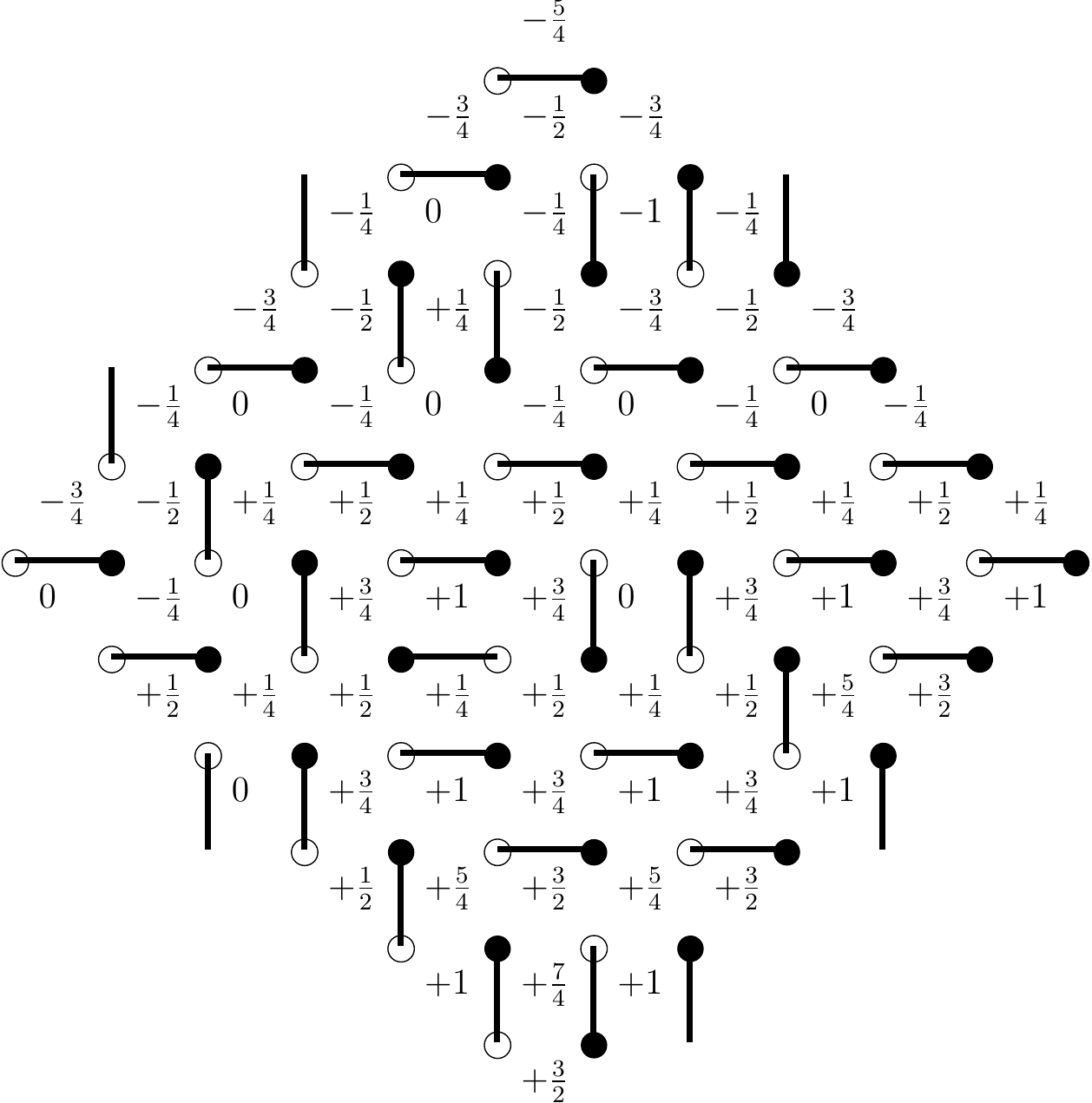}
\caption{The height function, $h(\eta)\in \mathbb Z/4$, corresponding to the dimer configuration of Fig.\ref{fig:sub1}.
On the torus, the height function is additively multi-valued: it can change by an integer along a cycle with non-trivial winding.
In general, if a cycle wraps over the torus $n_1$ times in direction $\vec e_1$ and $n_2$ times in direction $\vec e_2$, the height function picks up an additive term $n_1
T_1+n_2 T_2$, for suitable constants, called {\it periods}. In the figure, $T_1=+2$ and $T_2=-1$.}
\label{fig:fig3}
\end{figure}

In the non-interacting case ($\lambda=0$), it has been proven \cite{LSW} that the curve 
$\Gamma$ tends in the scaling limit to a variant of ${\rm SLE}_8$. 
In the interacting case, we conjecture that $\Gamma$ converges to (a variant of) ${\rm SLE}_{\kappa'}$, with
$\kappa'$ the largest root of \eqref{eq:Akappa}.
The conjecture is based on the series of works \cite{MS1, MS2, MS3, MS4} (see notably \cite{MS4}), where
the authors, directly at the continuum level, give a meaning to the  notion of
solution of the differential equation 
\be \dot \gamma(t)= e^{i(2\p \phi(\gamma(t))+\theta)},\label{flowMS.bis}\ee
where $\phi$ is a massless Gaussian field with covariance \eqref{eq:varia}, and $\theta$ a real constant.
Note that the flow line $\gamma(t)$, solution of \eqref{flowMS.bis},
is such that $2\p \big[\phi(\gamma(t))-\phi(\gamma(0))\big]$ describes the net amount of winding of $\dot\gamma(t)$ around the circle between times $0$ and $t$.
In \cite{MS1, MS2, MS3, MS4} the authors define:  a continuous tree whose branches, roughly speaking, are the solutions of
\eqref{flowMS.bis} starting from any  point on the plane; and a space-filling curve  that traces the tree in the natural order.
They prove\footnote{In \cite{MS4} the authors use a different convention, namely,
they study the solution of $\dot\g(t)=e^{i(\varphi(\g(t))/\chi+\theta)}$, with $\varphi$ the massless Gaussian field with covariance $-\log|x-y|$, and they prove that it defines a space-filling SLE$_{\k'}$ process, 
with $\k'$ and $\chi$ related by $\chi=\sqrt{\k'}/2-2/\sqrt{\k'}$. Using the fact that $\varphi=\sqrt{(2\p^2/A(\l))}\phi$, we find that $\chi=1/\sqrt{2A(\l)}$ and we see that the relation between $\chi$ and $\k'$ implies \eqref{eq:Akappa}.} that the latter  defines a space-filling SLE$_{\k'}$ process, with $\k'$ the largest root of \eqref{eq:Akappa}: this 
space-filling curve is the continuous analogue of our space-filling discrete curve $\Gamma$, which leads to our conjecture.
In \cite{MS4}, it is also proved that the flow lines $\gamma$ define a branched version of an SLE$_{\k}$ process, with $\kappa$ the diffusion constant `dual' to $\kappa'$, i.e., $\kappa=16/{\kappa'}<4$: 
this branched process is the continuous analogue of the Temperley's CRSF; therefore, we also conjecture that the Temperley's CRSF tends in the scaling limit to the variant of SLE$_{\k}$ described in \cite{MS4}.

Let us mention that recently a similar conjecture has been formulated
(and successfully tested numerically) in \cite{KMSW} for the 6V
model that, as recalled above, can be equivalently described in the form of an interacting dimer model.
However, the discrete `space-filling' curve considered in \cite{KMSW}
is different from $\Gamma$. In particular, at the free-fermion point 
(that corresponds to $\lambda=0$ for dimers) the curve of \cite{KMSW}
converges to SLE$_{8+4\sqrt 3}$ rather than SLE$_8$.

\section{Generating function and fermionic representation}\label{section3}

In this section, we discuss the fermionic representation of the dimer model, which is the basis of the multiscale expansion 
used in the proof of Theorem \ref{th:k=k}.

\bigskip

The generating function of dimer correlations is defined as
\begin{eqnarray}
  \label{eq:genera}
  e^{\mathcal W_\L({A})}
  =\langle\prod_{e\in\Lambda}e^{A_e\openone_e}\rangle_{\L;\l,m},
\end{eqnarray}
where for each edge $e$, $A_e$ is a real number. 
In particular, the dimer-dimer correlation function in the thermodynamic and zero mass limits is 
 \begin{eqnarray}
   \label{eq:13}
G^{(0,2)}_{r,r'}(x,y):=\lim_{m\to0}\lim_{L\to\infty}\partial^2_{A_e,A_{e'}}\mathcal W_\L({A})|_{{A}=0}
 \end{eqnarray}
where $e,e'$ are the bonds of type $r,r'$ with $b(e)=x,b(e')=y$. Note that $G^{(0,2)}_{r,r'}(x,y)$ depends only on $x-y$ and its Fourier transform is defined via 
\be \label{eq:G02F}
   G^{(0,2)}_{r,r'}(x,0)=\int\limits_{[-\p,\p]^2}\frac{dp}{(2\p)^2}e^{i p x}\hat G^{(0,2)}_{r,r'}(p).\ee

As discussed in \cite[Sec. 2]{GMTlungo}, since we are working on the torus, the generating
function at finite volume can be rewritten as the sum of four Grassmann integrals (the
same way as the partition function of the non-interacting model on the
torus is written in Kasteleyn's theory as the sum of four determinants
or Pfaffians \cite{K1}). However (cf. \cite[Sec. 2.4]{GMTlungo}) since the mass $m$ is
removed \emph{after} the thermodynamic limit is taken, one can reduce to a single
Grassmann integral, the structure of which is discussed in the following subsections.

\subsection{Kasteleyn matrix and non-interacting model}

Let us start by defining the Kasteleyn matrix $K$: this is a square
matrix with rows/columns indexed by elements of $\L$. If the black
site $b_x$ of coordinate $x$ is not a neighbor of the white site $w_y$
of coordinate $y$, then we set $K(x,y)=0$. Otherwise, if $b_x,w_y$ are
the endpoints of an edge of type $r$, we set
 \begin{align}
   \label{eq:Kst}
K(x,y):=K_r \;t^{(m)}_{(b_x,w_y)},\quad
K_r=e^{i\frac\pi2(r-1)}   
 \end{align}
where $t^{(m)}_e$ was defined in \eqref{eq:tem}.
Given an edge $e$ such that $b(e)=x,w(e)=y$, we will write
\[
K(e):=K(x,y).
\]

The matrix $K$ is invertible as long as $m>0$ and can be explicitly
inverted in the Fourier basis. One has the following properties (these
properties were discussed in detail in \cite[Sec. 2 and
App. A]{GMTlungo}, except that there we used a different coordinate
system on $\mathbb Z^2$ (cf. Remark \ref{rem:mc1}), which explains why
e.g. \eqref{eq:72} below looks at first sight different from the
$m\to0$ limit of \cite[Eq. (2.18)]{GMTlungo}):
\begin{itemize}
\item if $m>0$ is fixed, then $K^{-1}(x,y)$ decays
  exponentially in $|x-y|$ on a characteristic length of order $1/m$,
  uniformly in $L$.
\item for $x,y$ fixed one has
  \begin{eqnarray}
    \label{eq:72}
    \lim_{m\to0}\lim_{L\to\infty}K^{-1}(x,y)=g(x-y):=\int_{[-\p,\pi]^2}
    \frac{dk}{(2\pi)^2}\frac{e^{-i k  (x-y)}}{\mu(k)}
  \end{eqnarray}
where
\begin{eqnarray}
  \label{eq:73}
  \mu(k)=1-e^{i(k_1+k_2)}+i e^{i k_1}-i e^{i k_2}.
\end{eqnarray}
\item $g(x-y)$ has the following long-distance behavior:
\begin{eqnarray}
  \label{eq:12}
g(x-y) \stackrel{|x-y|\to\infty}= \frac 1{2\pi}\left[\frac{1}{a(x-y)}+\frac{(-1)^{x-y}}{a^*(x-y)}\right]+O(|x-y|^{-2})
\end{eqnarray}
where
\begin{eqnarray}
  \label{eq:40}
a(z)=(1-i)z_1+(1+i)z_2.
\end{eqnarray}

\end{itemize}
The only zeros of $\mu(\cdot)$ on $[-\p,\pi]^2$ are $(k_1,k_2)=(0,0)=:p^+$ and
$(k_1,k_2)=(\pi,\pi)=:p^-$ and these are simple zeros:
\bea 
  \label{eq:10}
&&  \mu(k)=D_{\omega}(k-p^\omega)+O(|k-p^\omega|^2), \quad \omega=\pm,
\\
&& D_\omega(k)=(-i-\omega)k_1+(-i+\omega)k_2.\label{eq:10bis}
\eea 

In order to represent the generating function $\mathcal W_\L({A},0)$ as a Grassmann integral, 
we associate a
Grassmann variable $\psi^+_x$ with the black vertex indexed $x\in \L$, and a Grassmann variable 
$\psi^-_x$ with the white vertex indexed $x$.
Let for brevity $\psi$ denote the collection of Grassmann
variables $(\psi^+_x,\psi^-_x)_{x\in \L }$ with anti-periodic boundary conditions on $\mathbb T_L$, i.e., 
$\psi^\pm_{x+L\vec e_j}=-\psi^\pm_x$, for all $x\in\L$ and $j=1,2$. Let also 
$\EE_{0,\L}(\cdot)$ be the Grassmann Gaussian `measure' with propagator
\bea   \label{eq:61}
&& 
 \EE_{0,\L}(\psi^-_x\psi^+_y):= 
\frac{\int[\prod_{x\in\L}d\psi^+_xd\psi^-_x]e^{-(\psi^+, K_0\psi^-)}\psi^-_x\psi^+_y}{\det(K_0)}=g_\L(x-y),\qquad\phantom{\cdot}
\\
&& \EE_{0,\L}(\psi^-_x\psi^-_y)=\EE_{0,\L}(\psi^+_x\psi^+_y)=0.
\eea 
where $K_0$ is the Kasteleyn matrix with $m=0$ and anti-periodic boundary conditions,
\be g_\L(x)=K_0^{-1}(x,0)=\frac1{L^2}\sum_{k\in\mathcal D_\L}\frac{e^{-i k  x}}{\mu(k)},\ee
and 
\be \mathcal D_\L=\{k=(k_1,k_2): k_i=\frac{2\p}{L}(n_i+1/2),\ 0\le n_i\le L-1\}.\ee

We refer to \cite[Section 4]{GeM} for a few basic facts on Grassmann integration. Here let us just recall the fermionic Wick rule:
\begin{eqnarray}
  \label{eq:Wick}
  \EE_{0,\L}(\psi^{-}_{x_1}\psi^+_{y_1}\cdots \psi^{-}_{x_n}\psi^+_{y_n})=\det G_n(\underline x,\underline y), 
  \end{eqnarray}
where $G_n(\underline x,\underline y)$ is the $n\times n$ matrix with elements $[G_n(\underline x,\underline y)]_{ij}=g_\L(x_i-y_j)$,
$i,j=1,\ldots,n$. In the following, we shall denote by $\EE_0$ the weak limit of $\EE_{0,\L}$ as $L\to\infty$. 

Given a bond $e=(b,w)$, we let $E_e=K_0(e)\psi^+_{b}\psi^-_{w}$, so that 
\be
  \label{eq:S}
(\psi^+,
K_0\psi^-)=\sum_{x,y\in\L}\psi^+_xK_0(x,y)\psi^-_y=\sum_e E_e.
\ee
The non-interacting generating function, in the thermodynamic and zero mass limits, can 
be rewritten as 
\be \lim_{m\to 0}\lim_{L\to\infty}\mathcal W_\L({A})\big|_{\l=0} =\lim_{L\to\infty} W_\L({A})\big|_{\l=0},\label{eq:5.17}\ee
where
\be W_\L({A})\big|_{\l=0}=\log \frac{\int[\prod_{x\in\L}d\psi^+_xd\psi^-_x]e^{-\sum_e E_ee^{A_e}}}{\int[\prod_{x\in\L}d\psi^+_xd\psi^-_x]e^{-\sum_e E_e}}
\ee
is the Grassmann generating function. 
By differentiating with respect to $A_{e_1},\ldots, A_{e_k}$ and then setting ${A}={0}$, and by using 
\eqref{eq:Wick}, it is apparent that the multipoint dimer correlations can be all computed explicitly in terms of a suitable fermionic Wick rule.

For later reference, let us also set our conventions for the Fourier transform of Grassmann fields:  we let
\begin{eqnarray}
  \label{eq:F1}
  \hat \psi^\pm_k=\sum_{x\in\Lambda}e^{\mp i k x}\psi^\pm_x, \quad k\in \mathcal D_\L,
\end{eqnarray}
so that
\begin{eqnarray}
  \label{eq:F2}
  \psi^\pm_x=\frac1{L^2}\sum_{k\in\mathcal D_\L}e^{\pm i k x}\hat\psi^\pm_k.
\end{eqnarray}

\subsection{Decomposition of Dirac fields}

It is convenient,  for the comparison with the
relativistic model of Section \ref{sec:relativistico}, to decompose
the Grassmann fields $\psi^\pm_x$ into sums of so-called Dirac fields.
First of all, we let $\chi^+(\cdot)$ be a non-negative `cut-off
function', i.e. a smooth function on $[-\p,\pi]^2$ satisfying the
following: 
\begin{itemize}
\item letting $\chi^-(k)=\chi^+(k-(\pi,\pi))$, one
has $\chi^+(k)+\chi^-(k)=1$;
\item  $\chi^+(\cdot)$ is centered at the
origin, is even in $k$ and its support does not include
$p^{-}=(\pi,\pi)$.
\end{itemize}
 As discussed in \cite[App. C]{GMTlungo}, it is
technically convenient to assume that $\chi^+$ is in the Gevrey class
of order $2$ (which in particular implies that it is $C^\infty$).  For
definiteness, one should think of $\chi^+(\cdot)$ as of a suitably
smoothed version of the indicator ${\bf 1}_{\{|k_1|+|k_2|\le
  \pi\}}$. 

By using the {\it addition principle} for Grassmann integrals, see \cite[Prop.1]{GMTlungo}, we 
rewrite the field $\psi$ as the combination of two independent Grassmann fields $\psi_\o$, $\o\in\{\pm\}$, via the following
\begin{gather}
  \label{eq:14}
  \psi^\pm_x=\sum_{\omega=\pm}e^{i p^\omega x}\psi^\pm_{x,\omega}
=\psi^\pm_{x,+}+(-1)^x \psi^\pm_{x,-}
\end{gather}
(here and in the following, whenever $x=(x_1,x_2)\in \mathbb Z^2$, we
let $(-1)^x:=(-1)^{x_1+x_2}$).  The rewriting \eqref{eq:14} should be
meant as the statement that, for every function $f$ of the Grassmann
field $\psi$,
\be \EE_{0,\L}(f(\psi))=\tilde\EE_{0,\L}(\tilde f(\{\psi^\pm_{x,+},\psi^\pm_{x,-}\}_{x\in\L})),\ee
where $\tilde f(\{\psi^\pm_{x,+},\psi^\pm_{x,-}\}_{x\in\L})=
f(\{\psi^\pm_{x,+}+(-1)^x\psi^\pm_{x,-}\}_{x\in\L})$, and $\tilde\EE_{0,\L}$ is the Grassmann Gaussian integration on the $\psi_\o$ fields with propagator
\be
  \label{eq:16}
\tilde\EE_{0,\L}( \psi^-_{x,\o}\psi^+_{y,\o})=\frac1{L^2}\sum_{k\in\mathcal D_\L} \frac{e^{-i (k-p^\o)(x-y)}}{\mu(k)}\chi^\o(k)\ee
and $\tilde\EE_{0,\L}(\psi^-_{x,\o}\psi^+_{x,-\o})=0$.

\begin{Remark}
  Note that the effect of the cut-off function $\chi^\o$ is to restrict the integration close to the singularity $p^\o$; at large distances,  one has therefore
  \begin{multline}
    \label{eq:16bis}
    \tilde\EE_{0,\L}( \psi^-_{x,\o}\psi^+_{y,\o})\simeq\frac1{L^2}\sum_{k\in\mathcal D_\L} \frac{e^{-i k(x-y)}}{D_\o(k)}\chi^+(k)\\\simeq
\frac1{2\pi}\frac 1{(1-i\o)(x_1-y_1)+(1+i\o)(x_2-y_2)}\\
=\frac1{2\pi(1-i\o)}\frac1{(x_1-y_1)+i\o(x_2-y_2)},
  \end{multline}
with $D_\o(k)$ defined in \eqref{eq:10}. This will be used in the comparison between the dimer model and the continuum fermionic model of Section \ref{sec:relativistico}.
\end{Remark}

\subsection{Interacting model}
\label{sec:intmod}
As discussed in \cite[Sec. 2.3 and 2.4]{GMTlungo}, Eq. \eqref{eq:5.17} admits the following natural generalization to the interacting case: 
\be \lim_{m\to 0}\lim_{L\to\infty}\mathcal W_\L({A})=\lim_{L\to\infty} W_\L({A}),\label{eq:5.17_int}\ee
where 
\be W_\L({A})=\log \frac{\int[\prod_{x\in\L}d\psi^+_xd\psi^-_x]e^{V(\psi,{
    A})}}{\int   [\prod_{x\in\L}d\psi^+_x d\psi^-_x]e^{V(\psi,0)}}\label{eq:gen_int}
\ee
and
\be 
\label{eq:VA}
V(\psi,{
  A})=\sum_{\gamma}(-1)^{|\g|}\alpha^{|\g|-1}\prod_{e\in\gamma}E_e
e^{A_e}\;.\ee Here, $\alpha=\exp(\lambda)-1 $, $\g$ are collections of
parallel bonds in $\mathbb T_L$ belonging to the same horizontal or
vertical strip: $\g=(e_1,\ldots,e_k)$, with $k\ge 1$.  In Section \ref{sec1},
in order to derive the lattice Ward identities, we will generalize
$W_\Lambda(A)$ to a generating function $W_\L(A,\phi)$, with $\phi$
being a ``Grassmann external source'', with respect to which we will take derivatives.

Note that, at zero external source, and at lowest non-trivial order in $\l$, the Grassmann action takes the form
\begin{multline}
  \label{eq:19}
 V(\psi,0)=-(\psi^+,K_0\psi^-)+ \lambda V_4(\psi)+O(\lambda^2),\quad {\rm where}\\
V_4(\psi)=-2 \sum_{x\in\L}\left[
\psi^+_x\psi^-_x\psi^+_{x+(0,1)}\psi^-_{x-(1,0)}+\psi^+_x\psi^-_x\psi^+_{x+(1,0)}\psi^-_{x-(0,1)}\right].
\end{multline}
 For later reference, we introduce the symbol $\EE_{\lambda,\L}( \cdot)$ for 
the `interacting measure'
\begin{eqnarray}
  \label{eq:misura}
  \EE_{\lambda,\L}( f(\psi))=\frac{\int \big[
\prod_{x\in\L}d\psi^+_x d\psi^-_x
\big]e^{V(\psi,0)}f(\psi)}{\int \big[
\prod_{x\in\L}d\psi^+_x d\psi^-_x
\big]e^{V(\psi,0)}}
\end{eqnarray}
and we let $\EE_\l$ be the weak limit of $\EE_{\l,\L}$ as $L\to\infty$. 

We also define
\be 
\label{eq:calI}
\mathcal I_e=\sum_{\g:\g\ni e}(-1)^{|\g|}\alpha^{|\g|-1}\prod_{e'\in\gamma}E_{e'},\ee
which 
is the fermionic counterpart of the dimer operator $\openone_e$, in the sense that for instance
\begin{eqnarray}
  \label{eq:59} \lim_{m\to 0}\lim_{L\to\infty}
\langle \openone_e\rangle_{\lambda;\L,m}=\EE_\lambda(\II_e),
\end{eqnarray}
and 
\begin{eqnarray}
   \label{eq:23}
   \lim_{m\to 0}\lim_{L\to\infty}
\langle \openone_e\openone_{e'}\rangle_{\lambda;\L,m}=\EE_\l( \mathcal I_e\mathcal I_{e'}),
 \end{eqnarray}
{\it provided} $e,e'$ have different
orientations, or they have the same orientation but there is no
$\gamma$ made of parallel adjacent bonds that contains both. In order to prove \eqref{eq:59} and \eqref{eq:23}, it is enough to differentiate 
\eqref{eq:5.17_int} with respect to $A_e$ and $(A_e,A_{e'})$, and then set ${A}\equiv{0}$ (the possibility of exchanging the 
limits $\lim_{m\to0}\lim_{L\to\infty}$ with the derivatives with respect to ${A}$ follows from the uniform bounds on the generating function proved in \cite{GMTlungo}).

\begin{Remark}\label{rmk:4} For future reference, it is useful to re-express the Grassmann action, $V(\psi,0)$, and the Grassmann counterpart of the dimer observable, $\mathcal I_e$, in terms of Dirac fields: 
for this purpose, it is enough to plug \eqref{eq:14} into the definition of $E_e$, and then use \eqref{eq:VA} and \eqref{eq:calI}. If we denote by $E^r_x$, $r\in\{1,2,3,4\}$, the operator $E_e$
with $e$ the edge of type $r$ with black vertex at $x$, we find: 
\begin{eqnarray}
  \label{eq:27}
  E^r_x=K_r\left[
\sum_\omega
  \omega^{r-1}\psi^+_{x,\omega}\psi^-_{x+v_r,\omega}+(-1)^x\sum_\omega (-\omega)^{r-1}\psi^+_{x,\omega}\psi^-_{x+v_r,-\omega}
\right]
\end{eqnarray}
where $v_1=0,v_2=-(1,0),v_3=-(1,1),v_4=-(0,1)$. Using this expression, we see that the quadratic part of the action, $-(\psi^+,K_0\psi^-)$, reads
\be -(\psi^+,K_0\psi^-)=\sum_{x,r,\o}K_r\omega^{r-1}\left[
\psi^+_{x,\omega}\psi^-_{x+v_r,\omega}+(-1)^{x+r-1}\psi^+_{x,\omega}\psi^-_{x+v_r,-\omega}
\right],\ee
while the `interaction' $V_{int}(\psi):=V(\psi,0)+(\psi^+,K_0\psi^-)$ can be re-expressed, at dominant order, as
\be \label{eq:21}V_{int}(\psi)=-16\alpha  \sum_x\psi^+_{x,+}\psi^-_{x,+}\psi^+_{x,-}\psi^-_{x,-} +\text{h.o.},
\ee
Here `h.o.' denotes higher order terms, namely `non-local quartic terms\footnote{These are the terms quartic in the Grassmann fields,
where one or more fields are replaced by their discrete derivatives: they result from rewriting $\psi^-_{x+v_r,\o}=\psi^-_{x,\o}+(\nabla_{v_r}\psi^-)_{x,\o}$ in \eqref{eq:27}, where $\nabla_{v_r}$ is by definition the 
discrete derivative in direction $v_r$.}',
or terms of order $6$ or higher in the Grassmann fields. All these higher order terms are {\it irrelevant} in the RG sense, see \cite[Sect.5 and 6]{GMTlungo}.

Similarly, if $e$ is the edge of type $r$ with black vertex at $e$, the operator $\mathcal I_e$ can be rewritten as 
\be
\label{eq:calIappr}
\mathcal I_e = -K_r \left[
\sum_\omega
  \omega^{r-1}\psi^+_{x,\omega}\psi^-_{x,\omega}+(-1)^x\sum_\omega (-\omega)^{r-1}\psi^+_{x,\omega}\psi^-_{x,-\omega}\right]+\text{ h.o.},
\ee
where the higher order terms are either non-local quadratic terms, or terms of order 4 or higher in the Grassmann fields. 
\end{Remark}
\section{Ward Identities}\label{sec1}

A crucial role in the proof of Theorem \ref{th:k=k} is played by
lattice conservation laws, which induce exact relations (Ward
Identities) among the correlation functions.  The basic microscopic
conservation law inducing the Ward identities is the fact that the
number of dimers incident to a given vertex of $V_B/V_W$ is
identically equal to $1$.  

In order to derive these identities, we need to generalize the
generating function \eqref{eq:gen_int}, since the identities will
involve correlation functions that \emph{cannot} be obtained  as
derivatives of $W_\L(A)$ with respect to the variables $A_e$. Namely,
given a set of Grassmann variables $\{\phi^+_x,\phi^-_x\}_{x\in\L}$ and (as before) a set of real numbers $A_e$ associated to edges
$e$, we define \be W_\L({A},{\phi})=\log\frac{\int[\prod_{x\in\L}d\psi^+_xd\psi^-_x]
e^{V(\psi,{A})+(\psi^+,\phi^-)+(\phi^+,\psi^-)}}{\int [\prod_{x\in\L}d\psi^+_x d\psi^-_x]e^{V(\psi,0)}}\label{eq:gen_int_phi}.\ee
Here, $(\psi^+,\phi^-):=\sum_{x\in\L}\psi^+_x\phi^-_x$, and similarly for $(\phi^+,\psi^-)$.

\begin{Remark}
As any function of a finite number of Grassmann variables,   the function $W_\L(A,\phi)$ is a polynomial in $\{\phi^+_x,\phi^-_x\}_{x\in\L}$. If 
$W_\L(A,0)$ is the monomial that contains none of the $\phi$ variables, then we have
\begin{eqnarray}\label{eq:idqqq}
W_\L(A,0)=W_\L(A),\end{eqnarray}
with $W_\L(A)$ defined in \eqref{eq:gen_int}.

We will also need to take derivatives of $W_\L(A,\phi)$
w.r.t. $\phi^\pm_x$. By this we simply mean the following: if the
sum of monomials of  $W_\L(A,\phi)$ that contains $\phi^\sigma_x$ is
written as $\phi^\sigma_x f(\phi)$ (with $f$
a polynomial not containing $\phi^\sigma_x$) then
$\partial_{\phi^\sigma_x}W_\L(A,\phi):=f(\phi)$. Note that
$\partial_{\phi^\sigma_x}$ and $\partial_{\phi^{\sigma'}_{x'}}$
anti-commute.
\end{Remark}

The following identities hold:
\begin{Proposition}
\label{prop:W1}
  For every $x,y,z\in\Lambda$ we have
\bea \label{eq:der3}
&&\sum_{e:b(e)=x}\partial_{A_e}\partial_{\phi^-_z}\partial_{\phi^+_y}W_\L(0,0)+  \delta_{x=z}\partial_{\phi^-_x}\partial_{\phi^+_y}W_\L(0,0)=0\\
&&\sum_{e:w(e)=x}\partial_{A_e}\partial_{\phi^-_z}\partial_{\phi^+_y} W_\L(0,0)+  \delta_{x=y}\partial_{\phi^-_z}\partial_{\phi^+_x} W_\L(0,0)=0.
\eea
\end{Proposition}
\begin{proof}
The starting point in the derivation of \eqref{eq:der3} is a  ``local gauge covariance property'' of the
Grassmann generating function $W_\L({A},{{\phi}})$, i.e. an identity satisfied by $W_\L({A},{{\phi}})$ when the Grassmann variables $\psi^\sigma_x$ are multiplied by a  phase depending on $x$ and $\sigma$.

Grassmann integrals are known to satisfy the following: if
$\psi_1,\dots,\psi_n$ are Grassmann variables and $\bar
\psi_i(\psi)=\sum_{j\le n}a_{ij}\psi_j$ are linear combinations of the
$\psi$ variables, then
\begin{eqnarray}
  \label{eq:1}
  \int \prod_j d\psi_j f(\psi)=(\det a)^{-1}\int \prod_j d\psi_j f(\bar\psi(\psi)).
\end{eqnarray}
In our case,  consider the phase transformation $\psi^\pm_x\to
\bar \psi^\pm_x= e^{i\a^\pm_x}\psi^\pm_x$, 
with $\a_x^\pm\in\mathbb R$. 
Note that the combination
$E_ee^{A_e}$ appearing in the definition of $V(\psi,A)$, see \eqref{eq:VA}, transforms as follows under the phase transformation:
$E_ee^{A_e}\to E_ee^{i\a_e+A_e}$, where
$\a_e=\a^+_{b(e)}+\a^-_{w(e)}$. Therefore Eq. \eqref{eq:1}, together
with the fact that in our case \[\det a=e^{i\sum_{x\in\L}(\alpha^+_x+\alpha^-_x)},\] implies 
\be W_\L({ A},{{\phi}})=-i\sum_{x\in\L}(\a^+_x+\a^-_x)+W_\L({ A}+i{\alpha},{{\phi}}e^{i{\alpha}})\label{eq:gauge_cov}\ee
with $({{\phi}}e^{i{\alpha}})^\sigma_x=\phi^\sigma_x e^{i\a^{-\sigma}_x},\sigma=\pm$.
By repeatedly deriving this identity with respect to ${A}$ and ${\phi}$, and then setting ${A}={\phi}={0}$, we 
obtain a sequence of exact identities among the correlation functions, known as Ward Identities. 

We are particularly interested in \eqref{eq:der3}, relating the ``vertex
function'',  $ \partial_{A_e}\partial_{\phi^-_z}\partial_{\phi^+_y}W_\L(0,0)$
with the ``dressed propagator'' $\partial_{\phi^-_x}\partial_{\phi^+_y}W_\L(0,0) $. In order to get that,
we start by deriving \eqref{eq:gauge_cov} once with respect to $\a^+_x$ and then set ${{\a}}={0}$, thus obtaining:
\begin{eqnarray}
\label{eq:der1}
\sum_{e:b(e)= x}\partial_{A_e}W_\L({  A},{\phi})+\phi^-_x\partial_{\phi^-_x} W_\L({ A},{\phi})=1\;.  
\end{eqnarray}
Similarly, deriving with respect to $\a^-_x$ and setting
${\a}={{0}}$, we obtain: 
\begin{eqnarray}
\label{eq:der2}
\sum_{e:w(e)= x}\partial_{A_e} W_\L({A},{\phi})+\phi^+_x\partial_{\phi^+_x} W_\L({A},{\phi})=1\;.  
\end{eqnarray}
Next, we derive \eqref{eq:der1}, \eqref{eq:der2} w.r.t. 
$\phi^+_y$ and $\phi^-_z$ and we set ${\phi}={A}={0}$, thereby finding \eqref{eq:der3}.
\end{proof}

Equations \eqref{eq:der3} can be rewritten as
\bea 
\label{quasiward}
&&\sum_{e:b(e)=x}\EE_{\lambda,\L}( \mathcal I_e;\psi^-_y\psi^+_z) + \delta_{x=z}\EE_{\lambda,\L}( \psi^-_y\psi^+_x)=0
\\
&&\sum_{e:w(e)=x}\EE_{\lambda,\L}( \mathcal I_e;\psi^-_y\psi^+_z)+  \delta_{x=y}\EE_{\lambda,\L}(\psi^-_x\psi^+_z)=0,\label{quasiward2}
\eea
where the semicolon indicates truncated expectation (i.e., $\EE_{\l,\L}(A;B)=\EE_{\l,\L}(AB)-\EE_{\l,\L}(A)\EE_{\l,\L}(B)$). 

\begin{Remark}
If $\l=0$, one can check that \eqref{quasiward} reduces to the
statement that $\EE_{0,\Lambda}(\psi^-_x\psi^+_y)=K_0^{-1}(x,y)$
satisfies $K_0 K_0^{-1}=I$, i.e. the non-homogeneous discrete
Cauchy-Riemann equation.
\end{Remark}

If $e$ is the edge of type $r\in\{1,2,3,4\}$ with $b(e)=x$, we let 
\begin{gather} G^{(2,1)}_r(x, y,z):=\EE_\lambda(
\mathcal I_e;\psi^-_y\psi^+_z)
\end{gather}
(recall that $\EE_\lambda$ is the $L\to\infty$ limit of $\EE_{\lambda,\L}$).
Similarly, we define
\begin{gather}G^{(2)}(x,y):=\EE_\lambda( \psi^-_x\psi^+_y).
\end{gather}
Note that both $G^{(2,1)}_r$ and $G^{(2)}$ are translationally invariant. Their Fourier transform is defined by 
\begin{gather}\label{eq:FG21}
\hat G^{(2)}(p)=\sum_{x} G^{(2)}(x,0)e^{i p x}\\
\hat G^{(2,1)}_r(k,p)=\sum_{x,z}e^{-i px-ikz}
G^{(2,1)}_r(x,0,z)
\end{gather}
so that
\begin{gather}
\label{eq:AFG21}
\int\limits_{[-\p,\p]^2} \frac{dp}{(2\pi)^2}e^{-ipx}\hat G^{(2)}(p)=G^{(2)}(x,0)\\
\int\limits_{[-\p,\p]^2} \frac{dp}{(2\pi)^2}\int\limits_{[-\p,\p]^2} \frac{dk}{(2\pi)^2}e^{ip x+ik z}\hat G^{(2,1)}_r(k,p)=G^{(2,1)}_r(x,0,z).
\end{gather}

From \eqref{quasiward}-\eqref{quasiward2} we finally deduce  our exact lattice Ward
identity of interest:
\begin{Proposition}[Lattice Ward Identities]
For every $x,y,z$, 
\be
\label{eq:WIreticolo}
\delta_{x=y}G^{(2)}(x,z)  -\delta_{x=z}G^{(2)}(y,x)=-\sum_{r=2}^4\nabla_{-v_r}G^{(2,1)}_r(x,y,z),\ee
where $(\nabla_n f)(x,y,z):=f(x+n,y,z)-f(x,y,z)$ is the (un-normalized) discrete derivative acting on the $x$ variable. 
\end{Proposition}
After Fourier transform, this identity reads
\be\label{eq:WIF}\hat G^{(2)}(k+p)-\hat G^{(2)}(k)=\sum_{r=2}^4(e^{-ipv_r}-1)\hat G^{(2,1)}_r(k,p),\ee

\section{The reference model}\label{sec:relativistico}

In order to prove Theorem \ref{th:k=k}, we intend to follow a logic analogous to
\cite{BMun,BMdr}. The general scheme is as follows: we introduce a
relativistic reference model whose correlation functions have the same
long distance behavior as our lattice model, provided the bare
parameters entering the definition of its action are properly
chosen. The relativistic model satisfies Ward Identities corresponding
to local chiral gauge invariance, which guarantees the validity of
exact scaling relations among the critical exponents of its
correlation functions (and, therefore, a posteriori, of the
correlation functions of the dimer model). In addition to the
relations between critical exponents, we can get exact relations among
the exponents and the {\it amplitudes} of the correlations, by
comparing the relativistic Ward Identities with the lattice ones,
which are the same at dominant order (asymptotically at large
distances, or for momenta close to the Fermi points).

In this section, we introduce the reference relativistic model and
recall its Ward Identities.  The reference model is nothing but the
formal scaling limit of the Grassmann integral of our dimer model,
properly regularized, thanks to the presence of: (1) a smooth
non-local density-density potential $v_0$, decaying on a fixed scale,
which sets the unit length (the smooth, rather than delta-like
potential, provides an ultraviolet cut-off on the fermionic
interaction), (2) an infrared regularization on the propagator, induced by 
the presence of a finite box of side $L$ and by the anti-periodic boundary conditions 
enforced on the Grassmann fields, (3) an ultraviolet
regularization on the propagator, which cuts off the momenta larger
than $2^N$, with $N\gg 1$, and (4) an underlying lattice of mesh $a$, which guarantees that the number of Grassmann variables is finite.
The limit $a\to 0$, followed by $N\to\infty$ and then $L\to\io$,  
is called the {\it limit of removed cut-offs}. Note
that in this limit the decay scale of $v_0$ is kept fixed: therefore,
even the limiting theory has an ultraviolet regularization, which
guarantees the finiteness of the bare parameters to be fixed.

Given $L>0$ and an integer $M$, we let $a=L/M$ and we define $\L=a\mathbb Z^2/L\mathbb Z^2$ to be the discrete torus of side $L$ and lattice mesh $a$. 
The reference model is defined in terms of a generating functional 
$\mathcal W^\L_N(J,\phi)$ parametrized by 
\begin{itemize}
\item four real constants $Z,Z^{(1)},Z^{(2)},\l_{\infty}$;
\item external sources $J$, where $J=\{J^{(j)}_{x,\o}\}^{j=1,2}_{\o=\pm,\, x\in\L}$ and $J^{(j)}_{x,\o}\in\mathbb R$;
\item external Grassmann sources $\phi$, where $\phi=\{\phi^\sigma_{x,\o}\}^{\s,\o=\pm}_{x\in\Lambda}$ and $\phi^\sigma_{x,\o}$ is a Grassmann variable.
\end{itemize}
The generating function is defined by 
\be\lb{vv1}
e^{\WW^\Lambda_{N}(J,\phi)} = \frac{\int\! P_Z^{[\le N]}(d\psi)
e^{\VV(\sqrt{Z}\psi) + \sum_{j=1}^2Z^{(j)}(J^{(j)},\,\rho^{(j)})+
Z[(\psi^{+},\phi^-)+(\phi^+, \psi^{-})]}}
{\int\! P_Z^{[\le N]}(d\psi)
e^{\VV(\sqrt{Z}\psi) }}
\;,\ee
where, letting $\int_\L dx:=a^2\sum_{x\in\L}$, we defined $(J^{(j)},\r^{(j)}):=\sum_{\o}\int_\Lambda d x\, J^{(j)}_{x,\o}\r^{(j)}_{x,\o}$,  with
\be\lb{rhodef}
\r^{(1)}_{x,\o} = \psi^+_{x,\o} \psi^{-}_{x,\o}\;,
\qquad 
\r^{(2)}_{x,\o} = \psi^+_{x,\o} \psi^{-}_{x,-\o}\;.
\ee
and 
$$(\psi^+,\phi^-):=\sum_{\o=\pm}\int_{\L} d x\, \psi^+_{x,\o}\phi^-_{x,\o},\quad  (\phi^+,\psi^-):=\sum_{\o=\pm}\int_{\L} dx\, \phi^+_{x,\o}\psi^-_{x,\o}\;,$$ 
Moreover, $P_Z^{[\le N]}(d\psi)$ is
the fermionic measure with propagator (satisfying anti-periodic boundary conditions over $\L$)
\be\lb{gth}
\frac{1}{Z} g^{[\le N]}_{R,\o}(x-y)=\frac{1}{Z}\frac{1}{L^2}\sum_{k\in\mathcal D}e^{-i k(x-y)}\frac{\chi_{N}(k)}{(-i-\o)k_1+(-i+\o)k_2}\;,
\ee
where  $\chi_{N}(k)=\chi(2^{-N}|k|)$, with $\chi:\mathbb R^+\to \mathbb R$ a $C^\infty$ cutoff function that is equal to 1, if its argument is smaller than 1, and 
equal to 0, if its argument is larger than 2, and $\mathcal
D=(2\pi/L)(\mathbb Z/M\mathbb Z+1/2)^2$ (we recall that $M=L/a$). 

\begin{Remark}
Note that, in the limit of removed cut-offs, the cut-off function $\c_N$ tends to 1, the Riemann sum over $k\in\mathcal D$ tends to the corresponding integral, 
so that \eqref{gth} reduces to $1/Z$ times the
inverse of the Dirac operator \be D^x_\o:=(1-i\o)\partial_{x_1}+(1+i\o)\partial_{x_2}.\label{Dxo}\ee
Also, compare \eqref{gth} with \eqref{eq:16bis}: asymptotically at large distances,
the lattice Grassmann fields $\psi^\pm_{x,\o}$ and the ones of the
continuous model have the same propagator (apart from the constant pre-factor $1/Z$).
\end{Remark}

Finally, the interaction is
\be\lb{gjhfk} \VV(\psi)=\frac{\l_\io}2 \sum_{\o=\pm}\int_\Lambda
dx \int_\Lambda dy\ v_0(x-y) \psi^+_{x,\o}
\psi^-_{x,\o}\psi^+_{y,-\o}\psi^-_{y,-\o}\;, \ee
where $v_0$ is a smooth rotationally invariant potential, exponentially decaying to zero at large distances, 
of the form
\be v_0(x)=\frac{1}{L^2}\sum_{p\in (2\pi/L)\mathbb Z^2} \hat v_0(p)
e^{ipx}\;, \label{eqv0}\ee
with $|\hat v_0(p)|\le C e^{-c |p|}$, for some constants $C$, $c$, and
$\hat v_0(0)=1$. Note the similarity between \eqref{gjhfk} and the dominant, `local', quartic interaction of the interacting dimer model, \eqref{eq:21}.

We shall use the following definitions\footnote{In the right sides of \eqref{eq:5.6}, the space label $x$, $y$, $z$ of the external fields $J^{(j)}_\o$, $\phi^\pm_\o$ 
should be actually interpreted as the points $x_a$, $y_a$, $z_a$ in $\L$ closest to $x$, $y$, $z$. Clearly, $\lim_{a\to 0}x_a=x$, etc. In the formulas, we drop the label $a$ just for lightness of notation.}:
\bea
&&G^{(2,1)}_{R,\o',\o}(x,y,z) =\lim_{L\to\io}\lim_{N\to\infty}\lim_{a\to 0} \frac{\partial^3}{\partial J_{x,\o'}^{(1)}\partial\phi^-_{z,\o}\partial\phi^+_{y,\o} }
\WW^\L_{N}(J,\phi)|_{J=\phi=0}\;,\nn\\
&&G^{(2)}_{R,\o}(x,y) =\lim_{L\to\io}\lim_{N\to\infty}\lim_{a\to 0} \frac{\partial^2}{ \partial\phi^-_{y,\o}
\partial\phi^+_{x,\o}} \WW^\L_{N}(J,\phi)|_{J=\phi=0}\;,\label{eq:5.6}\\
&&S^{(j,j)}_{R,\o,\o'}(x,y)=\lim_{L\to\io}\lim_{N\to\infty}\lim_{a\to 0}\frac{\partial^2}{\dpr J_{x,\o}^{(j)} \dpr
J_{y,\o'}^{(j)}}\WW^\L_{N}(J,\phi)|_{J=\phi=0}\;.\nn
\eea
The very existence of the limits in the right sides follows from the construction of the correlation functions of the 
reference model, performed, e.g., in \cite[Sect. 3 and 4]{BFM}. The Fourier transforms of the correlations in \eqref{eq:5.6} are defined as follows:
\bea
&&G^{(2,1)}_{R,\o',\o}(x,y,z)=\int_{\mathbb R^2} \frac{dk}{(2\pi)^2}\int_{\mathbb R^2}\frac{dp}{(2\pi)^2} e^{ipx -i(k+p)y
+ikz} \hat G^{(2,1)}_{R,\o',\o}(k,p)\;,\nn\\
&&G^{(2)}_{R,\o}(x,y) = \int_{\mathbb R^2}\frac{ dk}{(2\pi)^2} 
e^{-ik(x-y)}\hat G^{(2)}_{R,\o}(k)\;,\\
&&S^{(j,j)}_{R,\o,\o'}(x,y)=\int_{\mathbb R^2} \frac{dp}{(2\pi)^2} e^{ip(x-y)}
\hat S^{(j,j)}_{R,\o,\o'}(p)\;.\nn
\eea

In the limit of removed cut-offs, $\hat G^{(2,1)}_{R,\o',\o}$ and $\hat G^{(2)}_{R,\o}$
satisfy the following remarkable identities: 
for small $\l_{\io}$, if
$k$ and $k+p$ are different from $0$,
\be\label{h11}
\frac{Z}{Z^{(1)}}\sum_{\omega'=\pm}D_{\omega'}(p)\hat G^{(2,1)}_{R,\omega',\o}(k,p)=
 \hat F(p) [\hat G^{(2)}_{R,\o}(k) - \hat G^{(2)}_{R,\o}(k+p)]\;,\ee
where $D_\omega(p)$ was defined in \eqref{eq:10bis}, 
and 
\begin{eqnarray}
  \label{eq:Ap}
 \hat F(p)=\frac1{1-\tau \hat v_0(p)}, \quad \tau=-\frac{\lambda_\infty}{8\pi}.
\end{eqnarray}

Equations \pref{h11} is a Ward Identity, derived by a local chiral
gauge transformation of the Grassmann field $\psi$, i.e., by the
$\o$-dependent phase transformation $\psi^\pm_{x,\o}\to e^{\pm
  i\a_\o(x)}\psi^\pm_{x,\o}$.  The fact that $\hat F(p)$ is not equal
to $1$ is a manifestation of the {\it anomalies} in quantum field
theory. Finally, the linearity of $1/\hat F(p)$ in terms of $\l_{\io}$
is a property called {\it anomaly non-renormalization}.
Eq.\eqref{h11} was proved in \cite{M07}, see also \cite{BFM} and the
comments after Prop.12 in \cite{GMTlungo}.  A sketch of its proof is
also discussed, for completeness, in Appendix \ref{app3}.

Similarly, the density-density correlations $S^{(1,1)}_{R,\o,\o'}(x,y)$ satisfy the following identities (see, again, Appendix \ref{app3}): 
\bea\lb{tt}
&&D_\o(p) \hat S^{(1,1)}_{R,\o,\o}(p) - \t\ \hat v_0(p) D_{-\o}(p) \hat
S^{(1,1)}_{R,-\o,\o}(p) + \frac{(Z^{(1)})^2}{8\pi  Z^2} D_{-\o}(p)=0\;,\nn\\
&&D_{-\o}(p) \hat S^{(1,1)}_{R,-\o,\o}(p) - \t\  \hat v_0(p) D_{\o}(p) \hat
S^{(1,1)}_{R,\o,\o}(p) =0\;,
\eea
which imply:
\bea
\label{goo}
&&\hat S^{(1,1)}_{R,\o,\o}(p) = -\frac{1}{Z^2}\frac{(Z^{(1)})^2}{8\p  (1-\t^2\hat v_0(p)^2)}
\frac{D_{-\o}(p)}{D_\o(p)}\;.
\eea

\section{Proof of Theorem \ref{th:k=k}}\label{sec8}

In order to transfer the information encoded in the Ward Identities of the reference model to the correlation functions of the dimer model, 
we choose the parameters $Z, Z^{(1)}, Z^{(2)}, \l_{\infty}$ of the
reference model so that its correlations are asymptotically the same as those associated with the 
lattice Grassmann generating function \eqref{eq:gen_int}, in the limit of small momenta/large distances. More precisely, the parameters of the reference model can be fixed so that the 
following asymptotic relation among the correlation functions of the two models is valid. 

\begin{Lemma}\lb{lm1} Given $\l$ small enough, there are constants $Z$, $Z^{(1)}$, $Z^{(2)}$,
$\l_\io$, depending analytically on $\l$, such that, if $|p|\le 1$,
\bea\lb{h10}
&&\hat G^{(0,2)}_{r,r'}(p) = K_rK_{r'} \sum_{\o=\pm}\o^{r+{r'}}\hat S^{(1,1)}_{R,\o,\o}(p)+R'_{r,r'}(p)\;,\\
&&\hat G^{(0,2)}_{r,r'}(p^-+p) = K_rK_{r'} \sum_{\o=\pm}(-1)^{r'-1}\o^{r+{r'}}\hat S^{(2,2)}_{R,\o,-\o}(p)+R''_{r,r'}(p)\;,\quad{\phantom{\cdot}}
\label{h10aa}
\eea
where $R'_{r,r'}(p)$ and $R''_{r,r'}(p)$ are continuous in $p$, notably at $p=0$. Moreover, supposing that $0<\mathfrak c\le |p|,|k|,|k+p|\le
2\mathfrak c$ 
for some $\mathfrak c<1$,  then for any $0<\th<1$ one has
\bea\lb{h10a}
&&\hat G^{(2,1)}_r(k+ p^\o, p)
= -K_r\sum_{\o'=\pm}(\o')^{r-1}\hat G^{(2,1)}_{R,\o',\o}(k,p)[1+O(\mathfrak c^\th)]\;,\\
&&\hat G^{(2)}(k+p^\o) = \hat G^{(2)}_{R,\o}(k)[1+O(\mathfrak c^\th)]\;.\lb{h10ab}
\eea
\end{Lemma}

\begin{Remark}
At lowest non-trivial order in $\l$, the parameters of the reference model are: $Z=1+O(\l)$, $Z^{(1)}=1+O(\l)$, $Z^{(2)}=1+O(\l)$, and $\l_\infty=-16\lambda+O(\lambda^2)$. 
In order to obtain these values, at lowest non-trivial order, it is enough to fix the parameters in the bare Grassmann actions of the reference and dimer models, so that the 
corresponding expressions match at dominant order: compare, e.g., \eqref{gth} with \eqref{eq:16bis} and \eqref{gjhfk} with \eqref{eq:21}.
\end{Remark}

Lemma \ref{lm1} is a restatement (using new notations) and a slight extension of \cite[Prop.12]{GMTlungo}. The comparison between Lemma \ref{lm1}
and \cite[Prop.12]{GMTlungo} is discussed in the next section. Let us now discuss its implications, in particular let us show how to use it in order to prove Theorem \ref{th:k=k}.

\bigskip

By combining \pref{h10a}-\eqref{h10ab} with \pref{eq:WIF}, and using the fact that $\sum_{r=2}^4K_r(ip\cdot v_r)(\o')^{r-1}=-D_{\o'}(p)$,
we obtain: 
\bea\lb{xa1}
&& \sum_{\o'=\pm}D_{\o'}(p)\hat G^{(2,1)}_{R,\o',\o}(k,p) 
= \left[ \hat G^{(2)}_{R,\o}(k) - \hat
G^{(2)}_{R,\o}(k+p) \right ] [1+O(\mathfrak c^\th)],\qquad {\phantom{\cdot}}
\eea
that, if compared with \eqref{h11} and recalling that $\hat v_0(p)=1+O(p)$, gives
\be \frac{Z^{(1)}}{(1-\t)Z}=1\;.\label{z1tau}\ee
Moreover, by combining \eqref{goo} with \eqref{h10} we obtain:
\be\lb{h100}
\hat G^{(0,2)}_{r,r'}(p) = - \frac{K_rK_{r'}(Z^{(1)})^2}{8\p Z^2(1-\t^2)}\sum_{\o=\pm}\o^{r+{r'}}\frac{D_{-\o}(p)}{D_\o(p)}+\tilde R_{r,r'}(p)\;,
\ee
where $\tilde R_{r,r'}$ is continuous in $p$, in particular at $p=0$. By using \eqref{z1tau}, we can rewrite this equation as
\be\lb{h101}
\hat G^{(0,2)}_{r,r'}(p) = -\frac{K_rK_{r'}}{8\p}\frac{1-\t}{1+\t} \sum_{\o=\pm}\o^{r+{r'}}\frac{D_{-\o}(p)}{D_\o(p)}+\tilde R_{r,r'}(p)\;.
\ee
Taking for instance $r=r'=1$ this implies that
\begin{eqnarray}
  \label{eq:rr1}
 \hat G^{(0,2)}_{1,1}(p) =-\frac{1}{2\pi}\frac{1-\tau}{1+\tau}\frac{p_1 p_2}{p_1^2+p_2^2}+\tilde R_{1,1}(p).
\end{eqnarray}
This
should be compared with the Fourier transform of \eqref{eq:28}. 
Indeed, recall from the definition \eqref{eq:13} that \[G^{(0,2)}_{1,1}(x)=\langle
\openone_e;\openone_{e'}\rangle_\lambda,\] with $e,e'$ two horizontal 
edges of type 1, with black sites $b(e),b(e')$ of coordinates $x$ and $0$ respectively, and 
white sites of coordinates $w(e), w(e')$ also of coordinates $x$ and $0$.
Therefore, \eqref{eq:rr1} must also equal
\begin{eqnarray}
  \label{eq:must}
  \sum_{x\in\mathbb Z^2} e^{-ip x}\langle \openone_e;\openone_{e'}\rangle_\lambda.
\end{eqnarray}
In the computation of the Fourier transform for small values of $p$,
both the error term $R_{1,1}(x)$ in \eqref{eq:28} as well as the
oscillating term proportional to $B(\lambda)$  give a
contribution that is continuous in $p$ for $p$ in a neighborhood of $0$.  In
contrast, the term proportional to $A(\lambda)$ has a Fourier
transform that is not continuous at $p=0$. Namely,
\begin{eqnarray}
\label{eq:rsi}
  &&\sum_{x\in\mathbb Z^2} e^{-ip x}\langle \openone_e;\openone_{e'}\rangle_\lambda=\\
  &&\qquad =  \frac{A(\l)}{2\p^2}\sum_{x\ne 0} e^{-i p x}\text{Re}\left[\frac1{((x_1+x_2)+i(x_2-x_1))^2}\right]+\bar R(p)\nonumber\\
  &&\qquad = -\frac{A(\l)}{2\pi}\frac{p_1p_2}{p_1^2+p_2^2}+\bar{\bar R}(p)
\end{eqnarray}
where $\bar R(p)$ and $\bar{\bar R}(p)$ are continuous in $p$, notably
at $p=0$ (for a proof of the second identity, see, e.g., Appendix
\ref{app:residui}).  By comparing this expression with \eqref{eq:rr1}
we get \be A(\l)=\frac{1-\t}{1+\t}\;.\label{eq:Atau}\ee

On the other hand, the right side of this equation, $(1-\t)/(1+\t)$, coincides with the critical exponent of the correlation $S^{(2,2)}_{R,\o,-\o}$. In fact, 
it was proved in \cite[Section 4.2]{BFM}, see in particular \cite[Eqs.(4.24),(4.26)]{BFM}, that 
\be S^{(2,2)}_{R,\o,\o'}(x,y)=C\frac{\d_{\o,-\o'}}{|x-y|^{2\nu}}+\bar R_{\o,\o'}(x,y),\label{eq:8.13}\ee
where: (i) for some $0<\th<1/2$ and a suitable constant $c_\th>0$, $|\bar R_{\o,\o'}(x,y)|\le c_\th|x-y|^{-3+\th}$; (ii) 
$C$ and $\nu$ are analytic functions of $\l_{\infty}$ and
\be \n=\frac{1-\t}{1+\t}.\label{eq:nutau}\ee
For completeness, we reproduce a sketch of the proof of this identity in 
Appendix \ref{app3}. 
 
By comparing \eqref{eq:28} with \eqref{h10aa},\eqref{eq:8.13}, we recognize that the critical exponent in \eqref{eq:8.13} is the same as the one in the second line of \eqref{eq:28}. Therefore, by combining \eqref{eq:Atau} with \eqref{eq:nutau}, we obtain the statement of Theorem \ref{th:k=k}, as desired. \qed

\section{Comparison of Lemma \ref{lm1} with \cite[Prop.12]{GMTlungo}}
\label{sec9}

As discussed after its statement, Lemma \ref{lm1} is essentially a restatement \cite[Prop.12]{GMTlungo} in the notation of this paper. 
More precisely, \eqref{h10} and \eqref{h10aa} are 
obtained by Fourier transforming \cite[Eqs.(6.91),(6.93)]{GMTlungo} with respect to $x-y$ and then by using 
\cite[Eq.(6.90)]{GMTlungo}, see below for more details. On the contrary, \eqref{h10a} and \eqref{h10ab} are not part of
\cite[Prop.12]{GMTlungo}, however, their proof follows by the same discussion as the one after \cite[Prop.12]{GMTlungo}.

\medskip

In order to help the reader recognizing that 
\eqref{h10} and \eqref{h10aa} are direct consequences of 
\cite[Eqs.(6.90),(6.91),(6.93)]{GMTlungo}, let us discuss more precisely the connection between the notations used here 
and in \cite{GMTlungo}. In order to avoid confusion, the parameters, fields and coordinates of the models in \cite{GMTlungo} will be 
referred here by adding an extra tilde on the symbols: we denote by 
$\tilde J^{(j)}_\omega(\tilde x)$ its source terms, by $\tilde v(\tilde x)$ its potential, 
by $\tilde \psi^\pm_{\tilde x,\omega}$ its Grassmann fields, etc. The coordinates $x$ used in this paper are related to those 
of \cite{GMTlungo}, denoted by $\tilde x$, by the formulas: $\tilde x(x)=(x_1+x_2,x_2-x_1)$ and its inverse $x(\tilde
x)=1/2(\tilde x_1-\tilde x_2,\tilde x_1+\tilde x_2)$. 

\medskip



The connection between the reference model in this paper and the one in \cite[Sect.6.3.2]{GMTlungo} is established by fixing:
\bea &&
\psi^\pm_{x,\omega}=\sqrt 2 \tilde\psi^\pm_{\tilde x(x),\omega},\qquad 
J^{(j)}_{x,\omega}=\tilde J^{(j)}_\omega(\tilde x(x)), \qquad v_0( x)=\tilde v(\tilde x(x)),\qquad \phantom{\cdot}\label{eq:9.2}
  \end{eqnarray}
and taking the same values for the parameters $\lambda_\infty,Z,Z^{(j)}$.
By fixing the parameters in this way, then the two relativistic models, in the limit of removed cut-offs, are exactly the same. 
The proof of this fact is straightforward, it is just a matter of tracking the normalization constants in the two cases correctly, and it is left to the reader. 
Note that, due to the factor $\sqrt2$ in the first of \eqref{eq:9.2}, the relativistic correlations $S^{(j,j)}_{R,\o,\o'}(x,y)$ defined in this paper are related 
to their analogues in \cite{GMTlungo} via the following:
\be S^{(j,j)}_{R,\o,\o'}(x,y)=4\tilde S^{(j,j)}_{R,\o,\o'}(\tilde x(x),\tilde x(y)).\label{eq:9.4}\ee
Now, by using this equation and the definition of dimer-dimer correlations, we immediately get \eqref{h10}-\eqref{h10aa}.
In fact, by direct inspection of \cite[Eq.(6.95)]{GMTlungo}, one sees that the 
dimer-dimer correlation in the notations of this paper is related to the one in \cite{GMTlungo} via
\begin{eqnarray}
  \label{eq:42}
  G^{(0,2)}_{r,r'}(x,y)=\left.\frac{\partial^2}{\partial \tilde J_{\hat
  x(x,r),j(r)}\partial \tilde J_{\hat y(0,r'),j(r')}}\tilde{\mathcal S}(\tilde{{\bf J}})\right|_{\tilde J=0},
\end{eqnarray}
where, given a coordinate $x $ and $r=1,2,3,4$, we let:  
$j(r)=1$ if $r=1,3$ (i.e. if the edge of type $r$ is horizontal), 
$j(r)=2$ if $r=2,4$ (i.e. if the edge is vertical), and 
\begin{eqnarray}
  \label{eq:43}
  \hat x(x,r)=
\left\{
  \begin{array}{ll}
    \tilde x(x)& \text{ if } r=1,2\\
\tilde x(x)-(1,0) & \text{ if } r=3\\
\tilde x(x)-(0,1) & \text{ if } r=4.
  \end{array}
\right.
\end{eqnarray}
Moreover, by using \cite[Eqs.(6.27)-(6.28), (6.91)--(6.93), (6.96)-(6.97)]{GMTlungo}, we can rewrite
\bea
  \label{eq:44}
  G^{(0,2)}_{r,r'}(x,y)&=&4 K_rK_{r'}
 \sum_{j=1}^2\sum_{\o=\pm} \Big[\o^{r+r'}\tilde S^{(1,1)}_{R;\o,\o}(\tilde x(x),\tilde x(y))\\
 &+&
 (-1)^{x-y}\o^{r-1}(-\o)^{r'-1}\tilde S^{(2,2)}_{R;\o,-\o}(\tilde x(x),\tilde x(y))\Big].\nonumber
\eea
By using \eqref{eq:9.4} and \cite[Eq.(6.90)]{GMTlungo}, and by taking Fourier transform at both sides, we finally get 
\eqref{h10} and \eqref{h10aa}.
The proof of \eqref{h10a} and \eqref{h10ab} goes along the same lines outlined after the statement of 
\cite[Prop.12]{GMTlungo} and we shall not belabor the details here.

\appendix
\section{Fourier singularities of the dimer-dimer correlation}
\label{app:residui}

In this section, we compute the dominant behavior of the Fourier transform $\hat G_{r,r'}^{(0,2)}(p)$ of the dimer-dimer correlation \eqref{eq:28}
close to the Fermi points, $p=0$ and $p=(\pi,\pi)$. The formulas derived here will be useful in the first order computation 
of $A(\l)$ and $\nu(\l)$, discussed in the next appendix.

\medskip

More precisely, we prove that, if $|p|\le 1$,
\bea 
\label{eq:sing1}
&& \hat G^{(0,2)}_{r,r'}(p)=-\frac{A(\lambda)}{4\pi}\text{Re}\left[K_r K_{r'}\frac{D_-(p)}{D_+(p)}
\right]+F_{r,r'}^+(p)\\
&&   \hat G^{(0,2)}_{r,r'}((\pi,\pi)+p)=\frac{t_{r,r'}}{2\pi}\frac{B(\lambda)}{2(\nu-1)}\Big[1-\Big(\frac{|p|}2\Big)^{2(\nu-1)}\frac{\G(2-\nu)}{\G(\n)}\Big]+
F^-_{r,r'}(p),\nonumber
\eea
where $\n=\n(\l)$, and $F^\pm_{r,r'}$ are functions that are continuous in $p$ in a neighborhood of $p=0$, uniformly in $\l$, for $\l$ small.

\subsection{Singularity at $p=0$} In order to prove the first of \eqref{eq:sing1}, we note that the discontinuous part of $\hat G^{(0,2)}_{r,r'}(p)$ 
at $p=0$ comes from the first term in the right side of \eqref{eq:28}, that 
is, if $|p|\le1$,
\be \hat G^{(0,2)}_{r,r'}(p)=-\frac{A(\l)}{2\p^2}\sum_{x\neq 0}e^{-ipx}\text{Re}\left[\frac{e^{i \pi/2(r+r')}}{((x_1+x_2)+i(x_2-x_1))^2}\right]+\tilde F^+_{r,r'}(p),
\label{eq:A.2}\ee
where $\tilde F^+_{r,r'}$ is continuous in $p$ in the region $|p|\le 1$, uniformly in $\l$. 
One can first rewrite, recalling \eqref{eq:12}-\eqref{eq:40} and the fact that $K_r=\exp(i\pi/2(r-1))$,
\begin{multline}
  \label{eq:Gnuova}
-\frac1{2\pi^2}\text{Re}\left[\frac{e^{i \pi/2(r+r')}}{((x_1+x_2)+i(x_2-x_1))^2}
\right]= -K_r K_{r'}g(v_{r'}-x)g(x+v_r)\\-K_rK_{r'}((-1)^{v_r}+(-1)^{v_{r'}})\frac{(-1)^{x_1+x_2}}{8\pi^2|x|^2}+O(|x|^{-3}).
\end{multline}
where we used also
the fact that $(-1)^{v_r+v_{r'}}=+1$, if $K_rK_{r'}\in\mathbb R$, and $(-1)^{v_r+v_{r'}}=-1$, if $K_rK_{r'}\in i\mathbb R$.
By plugging \eqref{eq:Gnuova} into \eqref{eq:A.2}, we get 
\begin{eqnarray}
  \label{eq:Gatto}
  \hat G^{(0,2)}_{r,r'}(p)=-A(\lambda)K_rK_{r'}\int\limits_{[-\p,\p]^2}\frac{dk}{(2\pi)^2}\frac{e^{-i k(v_r+v_{r'})}}{\mu(k)\mu(k+p)}
+\bar F^+_{r,r'}(p),
\end{eqnarray}
where $\bar F^+_{r,r'}$ is continuous in $p$ in the region $|p|\le 1$, uniformly in $\l$. 

Now, the desired result follows from the explicit computation of the integral over $k$, which is summarized in the following proposition, formulated here in greater generality, for later convenience 
(the first of \eqref{eq:sing1} follows by an application of this proposition, with $(a,b)=-v_r- v_{r'}$).

\begin{Proposition}
\label{prop:Ibis}
Let $a,b\in \mathbb Z$. Let
\begin{eqnarray}
  \label{eq:I}
  I_{(a,b)}(p)=\int_{[-\p,\pi]^2} \frac{dk}{(2\pi)^2}\frac{e^{i k_1 a+i k_2 b}}{\mu(k)\mu(k+p)}.
\end{eqnarray}
Then, one has
\begin{multline}
  \label{eq:Ibis}
  I_{(a,b)}(p)=\frac 1{8\pi}\left[\frac{D_-(p)}{D_+(p)}+(-1)^{a+b}\frac{D_+(p)}{D_-(p)}\right]
+\frac i{8\pi}\left[
1-(-1)^{a+b}
\right]
\\
+\frac{(1-b)}{2\pi}\int_0^{2\pi}dk_1 e^{i k_1 a}\frac{(1+i e^{i k_1})^{b-2}}{(e^{i k_1}+i)^b}
\left[1_{\{b\le 0\}} 1_{\{\pi\le k_1\le 2\pi\}}-1_{\{b\ge 1\}}1_{\{0\le k_1\le\pi\}}\right]
\\+R_{a,b}(p)
\end{multline}
where $R_{a,b}(p)$ vanishes continuously for $p\to0$.
Moreover,
\begin{eqnarray}
  \label{eq:Isimmetrie}
  I_{(a,b)}(-p)=(-1)^{a+b}I_{(a,b)}(p)^*,\quad  I_{(a,b)}(-p)=I_{(a,b)}(p)+R'_{(a,b)}(p)
\end{eqnarray}
where $R'_{a,b}(p)$ vanishes continuously for $p\to0$.
\end{Proposition}
Checking \eqref{eq:Isimmetrie} is a simple exercise. 
The proof of \eqref{eq:Ibis} is lengthy but straightforward; we do not give details but only a few hints. One starts by rewriting, via the change of variables $z_j=e^{i k_j},j=1,2$,
\begin{eqnarray}
  I_{(a,b)}(p)=-\frac1{(2\pi)^2}\oint \frac{dz_1}{z_1}\oint\frac{dz_2}{z_2}\frac{z_1^a z_2^b}{\hat \mu(z_1,z_2)\hat\mu(z_1 e^{i p_1},z_2 e^{i p_2})},
\end{eqnarray}
with $\hat\mu(z_1,z_2)=1-z_1z_2+i z_1-i z_2$ and the integrals running over
$|z_1|=|z_2|=1$. The denominator, as a function of $z_2$, is a polynomial
with three distinct zeros (as long as $p\ne0$). The integral w.r.t. $z_2$
is performed via the residue theorem and the remaining integral over
$z_1=e^{i k_1}$ produces \eqref{eq:Ibis}. The only point that requires attention is that according to the value of $p$, some zeros can be inside or outside the integration curve $|z_2|=1$.

\subsection{Singularity at $p=(\pi,\pi)$}
\label{sec:pipi}

Let us now prove the second of \eqref{eq:sing1}. From \eqref{eq:28} we see that
\begin{eqnarray}
  \label{eq:sing2}
  \hat G^{(0,2)}_{r,r'}((\pi,\pi)+p)=t_{r,r'}\frac{B(\lambda)}{4\pi^2}\sum_{x\ne0} \frac{e^{-i x p}}{|x|^{2\nu}}
+\tilde F^-_{r,r'}(p)
\end{eqnarray}
where $\tilde F^-_{r,r'}$ is continuous in $p$ in the region $|p|\le 1$, uniformly in $\l$. 
We have
\bea
  \label{eq:89}
   \sum_{x\ne0} \frac{e^{-i x
       p}}{|x|^{2\nu}}&=&\int_{|x|\ge1}\frac{e^{-i x
       p}}{|x|^{2\nu}}dx+B_1(p)\\&=&2\pi
   |p|^{2(\nu-1)} \int_{|p|}^\infty \frac{J_0(\rho
     )}{\rho^{2\nu-1}}d\rho+B_1(p)\nonumber
\eea
where $J_0$ is the Bessel function with index $n=0$ and $B_1(p)$ is continuous in a neighborhood of $p=0$, uniformly in $\nu$, for $\n-1$ small. 
Now, by integration by parts, and recalling that $J_0'(\r)=-J_1(\r)$ and $J_0(0)=1$, we find
\bea &&\int_{|p|}^\infty \r^{1-2\n}J_0(\rho)d\r=\frac1{2(\n-1)}\Big[
J_0(|p|)|p|^{2(1-\n)}-\int_{|p|}^\infty\!\!\!\! J_1(\r)\r^{2(1-\n)}\Big]\\&&\quad =
\frac1{2(\n-1)}\Big[
|p|^{2(1-\n)}-\int_{0}^\infty J_1(\r)\r^{2(1-\n)}\Big]+\int_0^{|p|}\!\!\!d\r J_1(\r)\int_\r^{|p|}\!\!\!dx\, x^{1-2\n}.\nonumber
\eea
Recalling that $J_1(x)$ vanishes linearly at $x=0$, we find that the last term is of order $O(|p|^{4-2\n})$, as $p\to 0$,
uniformly in $\n$ for $\n-1$ small. Moreover, the integral $\int_{0}^\infty J_1(\r)\r^{2(1-\n)}$ is explicitly known, see \cite[formula 6.561(14)]{GR}: 
\be \int_{0}^\infty J_1(\r)\r^{2(1-\n)}=2^{2(1-\n)}\frac{\G(2-\n)}{\G(\n)}.\ee
By plugging this formula in \eqref{eq:sing2}-\eqref{eq:89} we obtain, as desired, the second of \eqref{eq:sing1} .

\section{First-order calculation}
\label{app:primordine}

Here we check that $A(\lambda)=\nu(\lambda)$ at first order
in $\lambda$ and, more precisely, that \eqref{eq:Kkl} holds.  From
\eqref{eq:13} and 
\eqref{eq:5.17_int}--\eqref{eq:VA} we have that, at first order in $\lambda$,
if $e$ and $e'$ are two edges such that $r(e)=r$, $b(e)=x$, and $r(e')=r'$, $b(e')=y$,
\bea 
  G^{(0,2)}_{r,r'}(x,y)&=&\EE_0(E_e;E_{e'})-\lambda\EE_0(\mathcal
  I^{(1)}_e;E_{e'})-\lambda \EE_0(E_e;\mathcal
  I^{(1)}_{e'})\nonumber\\
&+& \lambda\EE_0(E_e;E_{e'}; V_4)+O(\lambda^2)  \label{eq:45}
\eea
where $\EE_0$ is the Gaussian Grassmann integration with propagator \eqref{eq:61}, in the limit $L\to\infty$, and the semicolon indicates truncated expectation. 
Moreover, given an edge $e_0$, we defined $\mathcal I^{(1)}_{e_0}:=E_{e_0}
(E_{e_1}+E_{e_2})$, with $e_1,e_2$ the two edges parallel to $e_0$ and
at a distance $1$ from it, while  $V_4$ was defined in \eqref{eq:19}.
By computing all the contributions to \eqref{eq:45} and by taking Fourier transform, we obtain (see below for the
details of the computation):
\bea && \hat G^{(0,2)}_{r,r'}(p)=   -(1-2\lambda)K_r
  K_{r'}\int \frac{dk}{(2\pi )^2}\frac{e^{-i k v_r-i(k+p)v_{r'}}}{\mu(k)\mu(k+p)} \nonumber\\
  && - 2\lambda K_r K_{r'}\int\frac{dk}{(2\pi)^2}\int\frac{dk'}{(2\pi)^2}
\frac{W(k,k',p)e^{-i v_r(k'-p)-i v_{r'}(k+p)}}{\mu(k)\mu(k+p)\mu(k')\mu(k'-p)}\label{eq:B.2}\\
&& + \frac{\lambda}{4}\int\frac{dk}{(2\pi)^2}\frac{
K_{r'}e^{-i
    v_{r'}(k+p)}f_r(k,p)+K_{r}e^{-i
    v_{r}(k+p)}f_{r'}(k,p)}{\mu(k)\mu(k+p)}+O(\lambda^2),\nonumber\eea
where the integrals over $k$ and $k'$ are taken over the Brillouin zone $[-\p,\p]^2$, and 
we recall that $v_1=(0,0),v_2=-(1,0),v_3=-(1,1),v_4=-(0,1)$, while $\mu(k)$ was defined in \eqref{eq:73}.
Moreover, 
\bea && W(k,k',p)=e^{i(k_1'+k_2'-p_2)}+e^{i(k_1'+k_2'-p_1)} +e^{i(k_1+k_2+p_2)} +e^{i(k_1+k_2+p_1)}\nonumber \\
&&\quad  -e^{i(k_2+k'_1+p_2)}-e^{i(k_1+k'_2+p_1)}-e^{i(k_2'+k_1-p_2)}-e^{i(k_1'+k_2-p_1)},\label{eq:B.3}
\eea
while $f_r(k,p)$, for $r=1,\ldots,4$, is defined by
\bea 
&&\hskip-.5truecm f_1(k,p)=e^{i(k_1+k_2)}(e^{ip_1}+e^{ip_2})-2+ie^{ik_2}(1+e^{ip_2})-ie^{ik_1}(e^{ip_1}+1),\nonumber\\
&&\hskip-.5truecm f_2(k,p)=2e^{i(k_1+k_2)}-(e^{-ip_1}+e^{-ip_2})+ie^{ik_2}(e^{-ip_1}+1)-ie^{ik_1}(1+e^{-ip_2}),\nonumber\\
&&\hskip-.5truecm f_3(k,p)=e^{i(k_1+k_2)}(1+e^{ip_2})-(e^{-ip_1}+1)+ie^{ik_2}(e^{-ip_1}+e^{ip_2})-2ie^{ik_1},\nonumber\\
&&\hskip-.5truecm f_4(k,p)=e^{i(k_1+k_2)}(e^{ip_1}+1)-(1+e^{-ip_2})+2ie^{ik_2}-ie^{ik_1}(e^{ip_1}+e^{-ip_2}).\nonumber\eea
Observe for later convenience that
\begin{eqnarray}\label{eq:76}
W(k,k',0)=2(e^{i k_1}-e^{i k'_1})(e^{i k_2}-e^{i k'_2}),
\end{eqnarray}
and that 
\begin{eqnarray}
  \label{eq:78}
f_r(k,0)=-2\m(k).  
\end{eqnarray}
In the following, we shall first explain how to obtain \eqref{eq:B.2}, and then we will analyze its behavior close to $p=0$ and $p=(\pi,\pi)$,
so that, by using \eqref{eq:sing1}, we will be able to identify the values of $A(\l)$ and of $\nu(\l)$, at first non-trivial order in $\l$.

\subsection{Proof of \eqref{eq:B.2}}
We start from \eqref{eq:45}. Let us first look at the term $\EE_0(\mathcal
I^{(1)}_e;E_{e'})$. Given a quartic polynomial of the Grassmann
variables, of the form
$$W(\psi)=\sum_{x_1,\dots,x_4}a(x_1,\dots,x_4)\psi^+_{x_1}\psi^-_{x_2}\psi^+_{x_3}\psi^-_{x_4}, $$
we define its ``linearization'' $\overline W$ as:
\begin{multline}
  \label{eq:70}
  \overline W(\psi)=\sum_{x_1,\dots,x_4}a(x_1,\dots,x_4)[-\psi^+_{x_1}\psi^-_{x_2}g(x_4-x_3)-\psi^+_{x_3}\psi^-_{x_4}g(x_2-x_1)\\+\psi^+_{x_3}\psi^-_{x_2}g(x_4-x_1)+\psi^+_{x_1}\psi^-_{x_4}g(x_2-x_3)].
\end{multline}
Note that $\lis W$ is obtained from $W$ by `contracting' in all possible ways two out of the four Grassmann fields, the contraction 
corresponding to the selection of a pair of $\psi^+\psi^-$ fields, and by the replacement of the selected pair by its average with respect to $\EE_0(\cdot)$.
Due to the truncated expectation, one has then
 \begin{eqnarray}
    \label{eq:54}
 \EE_0(\mathcal
  I^{(1)}_e;E_{e'})=   \EE_0(\overline{\mathcal I}^{(1)}_e;E_{e'}).
  \end{eqnarray}
If, e.g., $e$ is of type $1$, 
\begin{eqnarray}
  \label{eq:57}
  \overline{\mathcal
  I}^{(1)}_e&=&g(v_1)(\psi^+_{x+(1,0)}\psi^-_{x-(0,1)}+\psi^+_{x+(0,1)}\psi^-_{x-(1,0)})+2g(v_3)\psi^+_x\psi^-_x\\
  &-&g(v_2)(\psi^+_{x+(0,1)}\psi^-_x
  + \psi^+_x\psi^-_{x-(0,1)})- g(v_4) (\psi^+_{x+(1,0)}\psi^-_{x}
  + \psi^+_x\psi^-_{x-(1,0)}).\nonumber
\eea
Since by symmetry all the edges $e$ have the same 
probability, $1/4$,  of
being occupied by a dimer, \eqref{eq:59} for $\lambda=0$ gives
$1/4=-\EE_0(E_e)$, so that
\be
  \label{eq:68}
g(v_r)=\EE_0(\psi^-_{x+v_r}\psi^+_x)=\frac1{4K_r}.
\ee
Therefore, if $e$ is of type $1$, 
\bea 
  \label{eq:58}
     \overline{\mathcal
  I}^{(1)}_e&=&\frac14(\psi^+_{x+(1,0)}\psi^-_{x-(0,1)}+\psi^+_{x+(0,1)}\psi^-_{x-(1,0)}-2\psi^+_x\psi^-_x)\\
  &+&\frac
i4 (\psi^+_{x+(0,1)}\psi^-_x
  + \psi^+_x\psi^-_{x-(0,1)}- \psi^+_{x+(1,0)}\psi^-_{x}
  - \psi^+_x\psi^-_{x-(1,0)})\nonumber\\
&=&\frac14\int\frac{dp}{(2\pi)^2}e^{i p
  x}\int\frac{dk}{(2\pi)^2}\hat\psi^+_{k+p}\hat\psi^-_kf_1(k,p)\nonumber
\eea
where in the last line we used the fact that $\psi^\pm_x=\int\frac{dk}{(2\p)^2}\hat \psi^\pm_ke^{\pm ikx}$,
as well as the definition of $f_1(k,p)$, see the equation after \eqref{eq:B.3}. 
A similar computation (details left to the reader) shows that, if $e$ is of type $r\in\{1,2,3,4\}$, 
\be  \overline{\mathcal
  I}^{(1)}_e=\frac14\int\frac{dp}{(2\pi)^2}e^{i p
  x}\int\frac{dk}{(2\pi)^2}\hat\psi^+_{k+p}\hat\psi^-_kf_r(k,p).\label{eq:B.11}\ee
Now, recall that 
\be
  \label{eq:64}
 E_{e'}=- K_{r'}\psi^+_y\psi^-_{y+v_{r'}}=-K_{r'}\int
  \frac{dp}{(2\pi)^2}e^{-ipy}\int \frac{dk}{(2\pi)^2}
\hat\psi^+_k\hat\psi^-_{k+p}e^{-i v_{r'}(k+p)}.
\ee
By using \eqref{eq:B.11}-\eqref{eq:64}, the Wick rule and the fact that $\EE_0(\hat \psi^-_k\hat\psi^+_{k'})=(2\p)^2$ $\delta(k-k')/\mu(k)$, 
we find that the Fourier transform of $-\lambda\EE_0(\mathcal
  I^{(1)}_e;E_{e'})$, computed at $p$, equals
\be
  \label{eq:69}
-\lambda  \sum_x e^{-i p x}\EE_0(\mathcal
  I^{(1)}_e;E_{e'})=   \lambda\frac{K_{r'}}4\int\frac{dk}{(2\pi)^2}\frac1{\mu(k)\mu(k+p)}e^{-i
    v_{r'}(k+p)}
f_r(k,p).\ee
\medskip

Next, we compute $\lambda\EE_0(E_e;E_{e'}; V_4)$. Among the
different ways of contracting the fields in the application of the
fermionic Wick rule to form connected diagrams, either two of the four fields
in a monomial 
of $ V_4$ are contracted among themselves, or they are all
contracted with fields in $E_e,E_{e'}$. 
One can check that the contributions of the former type, combined with the
zero-order diagram   $\EE_0(E_e;E_{e'})$, altogether give
\begin{eqnarray}
  \label{eq:66}
\bar\EE_0(E_e;E_{e'})+O(\lambda^2)
\end{eqnarray}
where $\bar \EE_0$ is the Grassmann Gaussian expectation defined in a way analogous to \eqref{eq:61}, 
with the difference that $-(\psi^+, K_0\psi^-)$ is 
replaced by $-(\psi^+, K_0\psi^-)+\lambda
\overline{ V}_4(\psi)$. Using
\eqref{eq:68} 
we find that 
\begin{eqnarray}
  \label{eq:67}
  \overline{ V}_4(\psi)=\sum_{x\in\L}\left[
  \psi^+_{x}\psi^-_{x-(1,1)}+i\psi^+_{x}\psi^-_{x-(0,1)}
  -\psi^+_x\psi^-_x- i\psi^+_x\psi^-_{x-(1,0)}\right]
\end{eqnarray}
which is nothing but $-(\psi^+,K_0\psi^-)$ itself. Therefore, 
\begin{eqnarray}
  \label{eq:52}
 \bar \EE_0(E_e;E_{e'})=K_r
  K_{r'}\bar \EE_0(\psi^+_{x}\psi^{-}_{x+v_r};\psi^+_y\psi^-_{y+v_{r'}}).
\\
=-\frac{K_r
  K_{r'}g(v_{r'}-x+y)g(v_r+x-y)}{(1+\lambda)^2}
\end{eqnarray}
Taking the Fourier transform, this gives
\begin{eqnarray}
  \label{eq:53}
  -(1-2\lambda)K_r
  K_{r'}\int \frac{dk}{(2\pi )^2}\frac{e^{-i k v_r-i(k+p)v_{r'}}}{\mu(k)\mu(k+p)}+O(\lambda^2).
\end{eqnarray}

Finally, we consider the contributions to $\lambda\EE_0(E_e;E_{e'}; V_4)$ from the diagrams where two fields in $V_4$ are contracted with
the two fields of $E_e$ and the remaining two are contracted with
$E_{e'}$. 
It is convenient to symmetrize $V_4$ by rewriting, after taking Fourier transform, 
\be \label{eq:71}
V_4(\psi)=-\frac12\int\frac{dp}{(2\pi)^2}\int\frac{dk}{(2\pi)^2}\int\frac{dk'}{(2\pi)^2}
\hat \psi^+_{k+p}\hat\psi^-_k\hat \psi^+_{k'-p}\hat
\psi^-_{k'}W(k,k',p),\ee
with $W$ as in \eqref{eq:B.3}. Now, by using \eqref{eq:71},
\eqref{eq:64},  and the analogous expression for $E_e$, we find that 
the Fourier transform of the sum of the contributions to $\l\EE_0(E_e;E_{e'}; V_4)$ from
connected diagrams, such that none of the fields of $V_4$ are contracted among themselves,
is equal to 
\begin{eqnarray}
  \label{eq:77}
 - 2\lambda K_r K_{r'}\int\frac{dk}{(2\pi)^2}\int\frac{dk'}{(2\pi)^2}
\frac{W(k,k',p)e^{-i v_r(k'-p)-i v_{r'}(k+p)}}{\mu(k)\mu(k+p)\mu(k')\mu(k'-p)}.
\end{eqnarray}
Now, recall that,  from \eqref{eq:45}, $\hat G^{(0,2)}_{r,r'}(p)$  at
first order in $\lambda$ is given by the sum of
\eqref{eq:53}, \eqref{eq:77}, \eqref{eq:69}, plus the term as in
\eqref{eq:69} but with $e,e'$ interchanged, which gives \eqref{eq:B.2}, as desired.

\subsection{Behavior for $p\sim 0$}
We want to identify the discontinuity of $\hat G^{(0,2)}_{r,r'}(p)$ at
$p=0$ and, by using \eqref{eq:sing1}, identify the pre-factor $A(\l)$ at lowest non-trivial order in $\l$.
We consider the case $r=r'=1$: any choice of $(r,r')$ is equally good for computing $A(\lambda)$,
and the case $r=r'=1$ is, possibly, the simplest. 

Let us first look at the first term in the right side of \eqref{eq:B.2}, which, for $r=r'=1$, is equal to 
\be-(1-2\lambda)\int \frac{dk}{(2\pi )^2}\frac{1}{\mu(k)\mu(k+p)},\label{eq:B.20bis}\ee
plus an error term that vanishes as $p\to 0$. The integral in \eqref{eq:B.20bis} can be evaluated by Proposition \ref{prop:Ibis}, and equals 
\be   \label{eq:85}
  -\frac{1-2\lambda}{4\pi}\text{Re}\left[
\frac{D_-(p)}{D_+(p)} \right]
\ee
plus terms that are continuous at $p=0$.

Next we look at the term in the second line of \eqref{eq:B.2}.  We can  can set $p=0$
in  $W(k,k',p)$, up to an error term that vanishes for $p\to0$. This
is because the integral diverges at like $(\log|p|)^2$, so that the
error term is $O(|p|(\log |p|)^2)$. Then, by using Proposition \ref{prop:Ibis}, 
the second line of \eqref{eq:B.2} for $r=r'=1$ reduces to
\begin{eqnarray}
\label{ree}
  -8 \lambda(I_{(1,1)}(p)I_{(0,0)}(p)-I_{(1,0)}(p)I_{(0,1)}(p))
\end{eqnarray}
plus terms that vanish as $p\to 0$. In order to evaluate this expression, we use \eqref{eq:Ibis},
which we rewrite as
\[
I_{(a,b)}(p)=\frac1{8\pi}\left[\frac{D_-(p)}{D_+(p)}+(-1)^{a+b}\frac{D_+(p)}{D_-(p)}
\right]+U_{(a,b)}+R_{(a,b)}(p).
\]
and we recall that $R_{(a,b)}$ vanishes continuously at
$p=0$. Omitting terms that are continuous at $p=0$, \eqref{ree} equals
\begin{gather}
  \label{eq:86}
  -\frac{2\lambda}\pi\text{Re}\left[
\frac{D_-(p)}{D_+(p)}(U_{(0,0)}+U_{(1,1)}-U_{(0,1)}-U_{(1,0)})\right].
\end{gather}
A simple computation shows that 
\[
U_{(0,0)}=\frac14\left(1-\frac2\pi\right),\quad U_{(1,1)}=0,\quad U_{(0,1)}=\frac
i{4\pi},\quad U_{(1,0)}=-\frac i{4\pi}
\]
so that \eqref{eq:86} reduces to 
\begin{gather}
  \label{eq:87}
  -\frac{\lambda}{2\pi}\left(1-\frac2\pi\right)\text{Re}\left[
\frac{D_-(p)}{D_+(p)}\right]
\end{gather}

Finally, it is easy to see that the term in the last line of \eqref{eq:B.2} is continuous at 
$p=0$. Indeed, as above, we can replace $f_1(k,p)$ by $f_1(k,0)$, up to an error term of order $O(|p|\log |p|)$. Now, recalling 
that $f_1(k,0)=-2\mu(k)$, we see that this term simplifies with the factor $\mu(k)$ in the denominator: after the simplification, we are left with 
an absolutely convergent integral. 

Altogether, from \eqref{eq:85} and \eqref{eq:87}
\begin{gather}
  \label{eq:84}
  \hat G^{(0,2)}_{r,r'}(p)=-\frac1{4\pi}\left(1-\frac{4\lambda}\pi\right) \text{Re}\left[
\frac{D_-(p)}{D_+(p)}\right]
\end{gather}
where we omitted terms that are either continuous at $p=0$ or are $O(\lambda^2)$.
In view of \eqref{eq:sing1}, we have
\begin{gather}
  \label{eq:88}
  A(\lambda)=1-\frac4\pi \lambda+O(\lambda^2).
\end{gather}
\subsection{Behavior for $p\sim (\pi,\pi)$}

In order to compute $\nu=\nu(\lambda)$ at first order, we use the second of \eqref{eq:sing1} that, if expanded at first order in $\l$,
and defining the coefficients $\n_1$ and $b_1$ via $\n(\l)=1+\n_1\l+O(\l^2)$ and $B(\l)\G(2-\nu(\l))/\G(\nu(\l))=1+b_1\l+O(\l^2)$, reads: 
\be  \hat G^{(0,2)}_{r,r'}((\pi,\pi)+q)=\frac{t_{r,r'}}{2\pi}\Big[-(1+b_1\l)\log\Big(\frac{|q|}2\Big)-\n_1\l\log^2\Big(\frac{|q|}2\Big)\Big],\label{eq:B.28}\ee
up to terms that are continuous at $q=0$, and/or of the order $O(\l^2)$. From this equation, it is apparent that $\n_1$ can be read from the 
most divergent contribution at order $O(\l)$, i.e., from the contribution to $\hat G^{(0,2)}_{r,r'}((\pi,\pi)+q)$ that diverges as $\log^2|q|$, as $q\to 0$, at first order in $\l$.

For simplicity, we consider again the case $r=r'=1$, in which case $t_{r,r'}=t_{1,1}=1$. We start from \eqref{eq:B.2} computed at $p=(\pi,\pi)+q$,
and we observe that neither the term in the first line nor the one in the third line 
give contribution to $\n_1$, since they diverge as $\log|q|$, as $q\to 0$. 

It remains to consider the second line of \eqref{eq:B.2},
which can be rewritten, up to terms that vanish as $q\to 0$, as
\be 4\lambda \int\frac{dk}{(2\pi)^2}\int\frac{dk'}{(2\pi)^2}
\frac{(e^{i k_1}-e^{i k'_1})(e^{i k_2}-e^{i
  k'_2})}{\mu(k)\mu(k+(\p,\p)+q)\mu(k')\mu(k'-(\p,\p)-q)},\label{eq:B.29}\ee
where we used that 
\[W(k,k',(\pi,\pi))=-W(k,k',0)=-2(e^{i k_1}-e^{i k'_1})(e^{i k_2}-e^{i
  k'_2}).\] The denominator has zeros when $k$ and $k'$ are close
either to $(0,0)$ or to $(\pi,\pi)$. However, when $k,k'$ are both
close  to $(0,0)$ or both close to $(\pi,\pi)$, the singularity is partially cancelled
by the fact that the numerator in \eqref{eq:B.29} vanishes there: as a consequence, the contributions to the integral 
in \eqref{eq:B.29} from the regions where $k,k'$ are both
close  to $(0,0)$ or both close to $(\pi,\pi)$ diverge less severely than $\log^2|q|$, as $q\to 0$. 

Therefore, the relevant contribution 
is from the region $k\sim 0,k'\sim(\pi,\pi)$, or vice-versa, in which case the numerator $(e^{i k_1}-e^{i k'_1})(e^{i k_2}-e^{i
  k'_2})$ in \eqref{eq:B.29} is equal to $4$, up to terms that, once
integrated in $k,k'$, give a contribution that is more regular than
$(\log |q|)^2$. We make a similar replacement in the denominator: consider, e.g., the case $k\sim 0,k'\sim(\pi,\pi)$,
and rewrite $k'= (\pi,\pi)+k''$, so that both $k$ and $k''$ are in a neighborhood of the origin. Then we can rewrite 
the denominator  $\mu(k)\mu(k+(\pi,\pi)+q)\mu(k')\mu(k'-(\pi,\pi)-q)$ in \eqref{eq:B.29} as its linearization, 
$D_+(k)D_-(k+q)D_-(k'')D_+(k''-q)$, plus a rest that, once
integrated in $k,k'$, give a contribution that is more regular than
$(\log |q|)^2$. We recall that $D_\pm(k)$ was defined in \eqref{eq:10bis}
and in particular $D_-(k)=-D_+(k)^*$. A similar argument can be repeated in the case that $k\sim (\p,\p),k'\sim 0$,
which gives exactly the same dominant contribution (and, therefore, can be accounted for by multiplying the result of the 
first case by an overall factor 2).

In conclusion, the dominant contribution (for
$q$ small) to \eqref{eq:B.29} is 
\begin{eqnarray}
  \label{eq:CA}
32\l \left|\int_{|k|< 1}\frac{dk}{(2\pi)^2}\frac1{D_+(k)D_-(k+q)}\right|^2.
\end{eqnarray}

Always at dominant order, we can restrict the integration to
$|k|>2|q|$ (the contribution from the complementary set is $O(1)$) and replace $D_+(k)D_-(k+q)$ by $-|D_+(k)|^2=-2|k|^2$ there.
Then, the integral \eqref{eq:CA} gives, at dominant order, 
\begin{eqnarray}
  \label{eq:CA2}
  32\lambda \left(\int_{2|q|<|k|<1}\frac{dk}{2(2\pi)^2}\frac1{|k|^2}\right)^2=\frac{2}{\pi^2}\l\big[\log(2 |q|)\big]^2.
\end{eqnarray}
By comparing this expression with the term of order $\log^2|q|$ in \eqref{eq:B.28}, and recalling that $t_{1,1}=1$, we conclude that
$\n_1=-4/\p$, that is 
\be 
\nu(\lambda)=1-\frac4{\p}\l+O(\lambda^2),\quad
\ee
as desired. 

\section{Ward Identities for the reference model}\label{app3}

In this section we sketch the proof of the Ward Identities \eqref{h11} and \eqref{tt}, and of the identity \eqref{eq:nutau} for the critical exponent $\nu$.
A full proof can be found in \cite{BFM} and references therein.

\subsection{Ward Identities}

The starting point in the derivation of the Ward Identities is a `chiral gauge transformation' of the Grassmann fields 
in the reference model's generating functional: that is, given $\bar \o\in\{\pm\}$, we perform 
the following change of variables:
\be
\psi^\pm_{x,\bar\o}\to e^{\pm i\a_{x,\bar\o}}
\psi^\pm_{x,\bar\o}\,,\quad \psi^\pm_{x,-\bar\o}\to 
\psi^\pm_{x,-\bar\o}\,,\label{eq:C.1}
\ee
in the numerator of the right side of \eqref{vv1}. The generating functional is invariant under this change of variables: 
therefore, the variation of the right side of \eqref{vv1} with respect to $\a_{\bar\o}:=\{\a_{x,\bar\o}\}_{x\in\L}$ is zero. In the following, we intend to compute 
this variation, take its derivative with respect to $\a_{\bar\o}$, and set $\a_{\bar\o}\equiv0$. The resulting identity can be thought of as the generating function for a hierarchy of Ward Identities. 

Note that the interaction $\mathcal V(\sqrt Z\,\psi)$ in the right side of \eqref{vv1} is invariant under the transformation \eqref{eq:C.1}, and so is the source term $Z^{(1)}(J^{(1)},\r^{(1)})$. Moreover, 
\bea && Z^{(2)}(J^{(2)},\r^{(2)}) \to Z^{(2)}\sum_{\o=\pm}\int_\L dx\, J^{(2)}_{x,\o}\r^{(2)}_{x,\o}\, e^{i\o\bar\o\a_{x,\bar\o}},\\
&&Z(\psi^+,\phi^-)\to Z\int_\L dx\big(\psi^+_{x,\bar\o}\phi^-_{x,\bar\o}e^{+i\a_{x,\bar\o}}+ \psi^+_{x,-\bar\o}\phi^-_{x,-\bar\o}\big),\\
&&Z(\phi^+,\psi^-)\to Z\int_\L dx\big(\phi^+_{x,\bar\o}\psi^-_{x,\bar\o}e^{-i\a_{x,\bar\o}}+ \phi^+_{x,-\bar\o}\psi^-_{x,-\bar\o}\big).\eea
Finally, and most importantly, the Gaussian integration is also affected by the chiral gauge transformation: in fact, $P_Z^{[\le N]}(d\psi)$
can be formally\footnote{The reason why \eqref{eq:C.2}-\eqref{eq:C.3} are {\it formal} is that $(\chi_N(k))^{-1}$ is infinite if $k$ is outside the support 
of $\c_N$. In order to make sense of these and the following formulas, one should introduce a function $\c_N^\e$ of full support, such that 
$\lim_{\e\to 0}\chi^\e_N(k)=\c_N(k)$. One should perform all computations keeping $\e$ fixed, and then send $\e\to0$ first, before the removal of all the other cut-off parameters. 
In this appendix, we neglect this issue: for more details, see \cite{BM02}.} written as
\be\lb{eq:C.2} P_Z^{[\le N]}(d\psi) = \frac1{\NN} \exp\Big\{ -Z\sum_\o\int_\L dx\, \psi^+_{x,\o} \big(D_{N}\psi^-_{\o}\big)(x) \Big\}
\prod_{\o=\pm}\prod_{x\in\L} d\psi^+_{x,\o} d\psi^-_{x,\o}\,,
\ee
where $\NN$ is a normalization constant and
\be
(D_{N}\psi^\pm_{\o})(x):=\mp\frac{1}{L^2}\sum_{k\in {\mathcal D}} e^{\pm i k x} (\chi_N(k))^{-1}D_\o(k)\hat\psi^\pm_{k,\o}\,.\label{eq:C.3}
\ee
Under the chiral gauge transformation, the `Grassmann measure' $$\prod_{\o=\pm}\prod_{x\in\L} d\psi^+_{x,\o} d\psi^-_{x,\o}$$ in the right side of \eqref{eq:C.2} is invariant, and so is 
the term $\int_\L dx\, \psi^+_{x,-\bar\o} \big(D_{N}\psi^-_{-\bar\o}\big)(x)$ at exponent, while
\bea && \int_\L dx\, \psi^+_{x,\bar\o} \big(D_{N}\psi^-_{\bar\o}\big)(x)\to 
\int_\L dx\, e^{+i\a_{x,\bar\o}}
\psi^+_{x,\bar\o} (D_{N}e^{-i\a_{\bar\o}} \psi^-_{\bar\o})(x)\\
&&=\int_\L\!\! dx\, \psi^+_{x,\bar\o} \big(D_{N}\psi^-_{\bar\o}\big)(x)+i\int_\L\!\! dx\, \a_{x,\bar\o}\Big[(D\r^{(1)}_{\bar\o})(x)+\d T_{\bar\o}(x)\Big]+O(\a_{\bar\o}^2),\nn
\eea
where
\bea &&(D \r^{(1)}_{\bar\o})(x)= -\frac1{L^4}\sum_{k,p\in\mathcal D}
e^{ipx} D_{\bar\o}(p)\hat\psi^+_{k+p,\bar\o} \hat\psi^-_{k,\bar\o}\,,\\
&&\delta T_{\bar\o}(x)=\frac{1}{L^4}\sum_{k,p\in\mathcal D}
e^{ipx} \hat\psi^+_{k+p,\bar\o} C_{\bar\o}(k+p,k)\hat\psi^-_{k,\bar\o}\,,\label{eq:C.9}
\eea
and 
\be\label{ccc} C_{\bar\o}(k+p,k) := [\c_{N}^{-1}(k)-1]
D_{\bar\o}(k) -[\c_{N}^{-1}(k+p)-1] D_{\bar\o}(k+p).\ee
Putting things together, and imposing that the variation of \eqref{vv1} vanishes at first order in $\a_{\bar\o}$, we find,
letting  \be \mathcal V(\psi,J,\phi):=\VV(\sqrt{Z}\psi) + \sum_{j}Z^{(j)}(J^{(j)},\,\rho^{(j)})+
Z[(\psi^{+},\phi^-)+(\phi^+, \psi^{-})],\ee
that
\bea && \int P_Z^{[\le N]}(d\psi)e^{\VV(\psi,J,\phi)}\big[Z^{(2)}\sum_\o \o\bar\o J^{(2)}_{x,\o}\r^{(2)}_{x,\o}+\label{eq:C.12}\\
&&+Z\psi^+_{x,\bar\o}\phi^-_{x,\bar\o}-Z\phi^+_{x,\bar\o}\psi^-_{x,\bar\o}
-Z(D\r^{(1)}_{\bar\o})(x)-Z\d T_{\bar\o}(x)\big]=0.\nn\eea
By taking derivatives w.r.t. $J$ and $\phi$, and then setting the external sources to zero, we generate a hierarchy of Ward Identities: for example, by 
deriving w.r.t. $\phi^-_{z,\o}$ and $\phi^+_{y,\o}$, and then setting the external sources to zero, we find
\bea && \d_{\o,\bar\o}\d_{x,z}Z^2\media{\psi^-_{y,\o}\psi^+_{x,\o}}_{L,N,a}-\d_{\o,\bar\o}\d_{x,y}Z^2\media{\psi^-_{x,\o}\psi^+_{z,\o}}_{L,N,a}\label{C.13}\\
&&-Z^3\media{(D\r^{(1)}_{\bar\o})(x);\psi^-_{y,\o}\psi^+_{z,\o}}_{L,N,a}-Z^3\media{\d T_{\bar\o}(x);\psi^-_{y,\o}\psi^+_{z,\o}}_{L,N,a}=0,\nn
\eea
where $\media{A(\psi)}_{L,N,a}:=\int P_Z^{[\le N]}(d\psi)e^{\VV(\sqrt Z \psi)}A(\psi)/\int P_Z^{[\le N]}(d\psi)e^{\VV(\sqrt Z \psi)}$, and the semicolon indicates 
truncated expectation (note that in the second line the truncated expectations are equal to the un-truncated ones, simply because $\media{(D\r^{(1)}_{\bar\o})(x)}_{L,N,a}=
\media{\d T_{\bar\o}(x)}_{L,N,a}=0$). After taking the Fourier transform and sending $L,N,a^{-1}\to\infty$, we get 
\be\label{eq:C.14}
\d_{\o,\bar\o}[\hat G^{(2)}_{R,\o}(k+p) - \hat G^{(2)}_{R,\o}(k)]+
\frac{Z}{Z^{(1)}}D_{\bar\o}(p)\hat G^{(2,1)}_{R,\bar\omega,\o}(k,p)-\frac{Z}{Z^{(1)}}\hat A_{\bar\o,\o}(k,p)=0,\ee
where
\be \hat A_{\bar\o,\o}(k,p)=\lim_{L,N,a^{-1}\to\infty}\frac{Z^{(1)}Z^2}{L^4}\sum_{q\in\mathcal D}
C_{\bar\o}(q+p,q)  \media{\hat\psi^+_{q+p,\bar\o} \hat\psi^-_{q,\bar\o};\hat \psi^-_{k+p,\o}\hat \psi^+_{k,\o}}_{L,N,a},\label{eq:C.15}\ee
and the limit in the right side is meant as the limit $a\to 0$ first, then $N\to\infty$, then $L\to\infty$. By summing \eqref{eq:C.14} over 
$\bar\o$, we will eventually get \eqref{h11}: of course, the subtle issue is to compute $\hat A_{\bar\o,\o}(k,p)$ (see below).

\begin{Remark} The term $\hat A_{\bar\o,\o}(k,p)$ is referred to as an anomaly term: formally, it would be zero if we replaced
the cut-off function $\c_N$ appearing in the definition of $C_\o(k+p,k)$ by $1$. Of course, there is an exchange of limits involved: we first
need to compute the expression
in the right side of \eqref{eq:C.15} under the limit sign, and then take the limit of removed cut-offs. 
In order to compute the right side of \eqref{eq:C.15}, we can use a multi-scale analysis similar to the one used in the computation of the free energy and 
the correlation functions of the interacting dimer model (see \cite{GMTlungo}) and of the reference model (see, e.g., \cite{BM02} or \cite{BFM}). Under 
RG iterations, the operator $\d T_\o$ appears to be marginal \cite{BM02,BM05}, with a non-trivial flow from the ultraviolet scale, down to the deep infrared: 
this induces a non trivial flow of the anomaly term, which converges to a finite value (its `infrared fixed point'), 
non-vanishing in the limit of removed cut-offs. The fact that the anomaly term does not vanish 
in this limit is visible already at first order in perturbation theory. 
\end{Remark}
In addition to the Ward Identity \eqref{eq:C.14} for the vertex function, we can get another identity for the density-density correlations, by deriving 
\eqref{eq:C.12} with respect to $J^{(1)}_{y,\o}$ and then setting the external sources to zero. In this case, after taking Fourier transform and the limit of removed cut-offs, we obtain 
\be D_{\bar \o}(p)\hat S^{(1,1)}_{R,\bar \o,\o}(p)=\hat B_{\bar\o,\o}(p),\label{eq:C.17} \ee
where
\be \hat B_{\bar\o,\o}(p)=\lim_{L,N,a^{-1}\to\infty}\frac{(Z^{(1)})^2}{L^6}\sum_{q,k\in\mathcal D}
C_{\bar\o}(q+p,q)  \media{\hat\psi^+_{q+p,\bar\o} \hat\psi^-_{q,\bar\o};\hat \psi^+_{k,\o}\hat \psi^-_{k-p,\o}}_{L,N,a}\label{eq:C.18}\ee
is an anomaly term. Its computation, sketched below, leads to \eqref{tt}.

\subsection{The anomaly terms} The anomaly terms have been computed in a series of previous works, starting from \cite{BM02,BM05}. 
In the presence of a non-local interaction potential, as the one considered here (recall \eqref{gjhfk}), it has a very explicit form: for instance,
$\hat A_{\o',\o}(k,p)$ is equal to 
an explicit pre-factor, {\it linear} in $\l_\infty$, times $\hat G^{(2,1)}_{R,-\o',\o}(k,p)$, see \eqref{c.21} below. The fact that the pre-factor is exactly 
linear in $\l_{\infty}$, i.e., that all its contributions beyond first order perturbation theory vanish in the removed cut-offs limit, is a phenomenon known as the {\it anomaly non-renormalization},
first proved in our context in \cite{M07}, see also \cite{BFM}. 
Here, we informally discuss the main ideas of the proof and explain how to compute the anomaly terms. 

For definiteness, let us consider the anomaly $\hat A_{\o',\o}(k,p)$
first: the expression in the right side of \eqref{eq:C.15}, under the limit sign, can be computed by a multi-scale analysis, similar to the one discussed in \cite{GMTlungo},
see \cite{BFM} for details. The outcome is that, if $k,p,k+p\neq0$, then $\hat A_{\o',\o}(k,p)$ is analytic in $\l_{\infty}$, uniformly in $N,L,a$. 
Note that, by the support properties of $C_{\bar\o}(q+p,q)$, in order for 
the summand not to vanish, both $q$ and $q+p$ must be `on scale $N$': that is, for any fixed $p$ and $N$ large enough, $2^{N-1}\le |q|, |q+p|\le 2^{N+2}$; therefore, the sum over $q$ 
in the right side of \eqref{eq:C.15} can be performed under this constraint. 

Thanks to the analyticity in $\l_\infty$, we can expand $\hat A_{\o',\o}(k,p)$ in series, thus finding 
\bea &&\hat A_{\o',\o}(k,p)= \lim_{L,N,a^{-1}\to\infty}\frac{Z^{(1)}Z^2}{L^4}\sum_{q\in\mathcal D}\sum_{n\ge 0}\frac1{n!}
C_{\bar\o}(q+p,q)\times\nn\\
&&\qquad \times  \langle\hat\psi^+_{q+p,\bar\o} \hat\psi^-_{q,\bar\o};\hat \psi^-_{k+p,\o}\hat \psi^+_{k,\o};\big[\VV(\sqrt{Z}\psi)\big]^{\!;n}\rangle^{0}_{L,N,a}\label{cc}\eea
where $\big[{\mathcal{V}}\big]^{\!;n}$ is a shorthand notation for 
$\underbrace{\VV;\VV;\, \cdots \,;\VV}_{\text{$n$ times}}$,
and $\media{(\cdot)}^{0}_{L,N,a}$ is the unperturbed integration, i.e.,
$\media{(\cdot)}^{0}_{L,N,a}:=\int P_Z^{[\le N]}(d\psi)(\cdot)$, whose expectations can be 
computed by the fermionic Wick rule.
Now, one can easily check that the $0$-th order term, i.e., the term with $n=0$ in the right side of \eqref{cc}, vanishes in the 
removed cut-off limit. The higher order terms can be of two types: either the two Grassmann fields 
$\hat\psi^+_{q+p,\bar\o} \hat\psi^-_{q,\bar\o}$ are contracted with (two of the four fields of) the same interaction term $\VV$, or with (the fields of) 
two different interaction terms. Correspondingly we can write: $\hat A_{\o',\o}(k,p)=(I)+(II)$, where, recalling that $\hat v_0(p)$ is the Fourier transform of the interaction potential, 
which is rotationally invariant and such that $\hat v_0(0)=1$,
\bea && (I)=\label{eq:ciao.ciao}\\
&&=\lim_{L,N,a^{-1}\to\infty}\frac{\l_\infty \hat v_0(p)}{L^6}\sum_{q\in\mathcal D}C_{\bar\o}(q+p,q)  Z^2
 \langle\hat\psi^+_{q+p,\bar\o} \hat\psi^-_{q,\bar\o};\hat\psi^+_{q,\bar\o} \hat\psi^-_{q+p,\bar\o}\rangle^0_{L,N,a}\nn\\
 &&\times\frac{1}{L^4} \sum_{q'\in\mathcal D}\sum_{n\ge 1}\frac{Z^{(1)}Z^2}{(n-1)!}
 \langle\hat\psi^+_{q'+p,-\bar\o} \hat\psi^-_{q',-\bar\o};
\hat \psi^-_{k+p,\o}\hat \psi^+_{k,\o};\big[\VV(\sqrt{Z}\psi)\big]^{\!;n-1}\rangle^{0}_{L,N,a}\nn\eea
and
\bea &&(II)=\label{eq:ciao}\\
&&=\lim_{L,N,a^{-1}\to\infty}\frac{1}{L^4}\sum_{q\in\mathcal D}C_{\bar\o}(q+p,q) 
 \langle\hat\psi^+_{q+p,\bar\o} \hat\psi^-_{q,\bar\o};\hat\psi^+_{q,\bar\o} \hat\psi^-_{q+p,\bar\o}\rangle^0_{L,N,a}\times\nn\\
 &&\times \sum_{n\ge 2}\frac{Z^{(1)}Z^2}{(n-2)!}
 \langle\frac{\partial\VV(\sqrt{Z}\psi)}{\partial\hat\psi^-_{q+p,\bar\o}};\frac{\partial\VV(\sqrt{Z}\psi)}{\partial\hat\psi^+_{q,\bar\o}};
\hat \psi^-_{k+p,\o}\hat \psi^+_{k,\o};\big[\VV(\sqrt{Z}\psi)\big]^{\!;n-2}\rangle^{0}_{L,N,a}\nn\eea
Let us focus on $(I)$ first. After summation over $n$, the  third line of \eqref{eq:ciao.ciao} can be explicitly re-written as 
$$\frac1{L^4}\sum_{q'\in\mathcal D}Z^{(1)}Z^2\media{\hat\psi^+_{q'+p,-\bar\o} \hat\psi^-_{q',-\bar\o};
\hat \psi^-_{k+p,\o}\hat \psi^+_{k,\o}}_{L,N,a},$$ 
whose removed cut-off limit equals $\hat G^{(2,1)}_{R,-\bar\o,\o}(k,p)$. Therefore,
\be (I)= \l_\infty\hat v_0(p)\hat H_\o(p) \hat G^{(2,1)}_{-\o',\o}(k,p),\ee
where 
\be \hat H_\o(p)=-\lim_{N\to\infty}\int_{\mathbb R^2}\frac{dq}{(2\p)^2}C_\o(q+p,q)\frac{\c_N(q)\c_N(q+p)}{D_\o(q)D_\o(q+p)}.\nn\ee
Recall that the integrand is non vanishing only if $q$ and $q+p$ are on scale $N$. Therefore, the integrand 
vanishes pointwise in the limit   $N\to\infty$. However, the limit of the integral is not zero: a straightforward, but very instructive, computation shows that 
\be \hat H_\o(p)=-\frac1{8\p}D_{-\o}(p), \label{eq:C.20}\ee
from which we infer that $(I)=-\frac{\l_\infty}{8\p}D_{-\o'}(p)\hat v_0(p)\hat G^{(2,1)}_{-\o',\o}(k,p)$. 

Let us now turn to $(II)$: the key fact is that all its contributions vanish in the removed cut-off limit. 
At a perturbative level, this can be checked by looking at the Feynman diagram expansion of the right side of 
\eqref{eq:ciao}: the constraint that the momenta $q,q+p$ are on scale $N$, while the external momenta $k,k+p$ are fixed, 
induces a `dimensional gain' in the dimensional estimate of the values of the diagrams, which makes them vanish in the removed cut-off limit,
at a speed at least $2^{-\theta N}$, for a suitable $\theta\in (0,1)$. The reader can easily check this fact for the lowest order diagrams. 
For a general, non-perturbative, proof, see \cite{BFM,M07}.

In conclusion,
\be \hat A_{\o',\o}(k,p)=-\frac{\l_\infty}{8\p}D_{-\o'}(p)\hat v_0(p)\hat G^{(2,1)}_{-\o',\o}(k,p).\label{c.21}\ee
Plugging this back into \eqref{eq:C.14} we get
\bea && 
\frac{Z}{Z^{(1)}}\Big[D_{\o}(p)\hat G^{(2,1)}_{R,\omega,\o}(k,p)+\frac{\l_\infty}{8\p}\hat v_0(p)D_{-\o}(p)\hat G^{(2,1)}_{-\o,\o}(k,p)\Big]=\nn\\
&&\hskip2.truecm =\hat G^{(2)}_{R,\o}(k) - \hat G^{(2)}_{R,\o}(k+p),\label{C.22}\\
&&\frac{Z}{Z^{(1)}}\Big[D_{-\o}(p)\hat G^{(2,1)}_{R,-\omega,\o}(k,p)+\frac{\l_\infty}{8\p}\hat v_0(p)D_{\o}(p)\hat G^{(2,1)}_{\o,\o}(k,p)\Big]=0.\nn\eea
and summing the two equations we get \eqref{h11}, as desired.

Similarly, we find that the anomaly term $\hat B_{\o',\o}(p)$ in \eqref{eq:C.18} is
\be \hat B_{\o',\o}(p)=\hat H_{\o'}(p)\Big[\frac{(Z^{(1)})^2}{Z^2}\d_{\o',\o}+\l_\infty\hat v_0(p)\hat S^{(1,1)}_{R,-\o',\o}(p)\Big],
\label{eq:C.23bis}\ee
with $\hat H_\o(p)$ as in \eqref{eq:C.20}. By plugging \eqref{eq:C.23bis} into \eqref{eq:C.17}, we immediately get \eqref{tt}.

\subsection{The critical exponent $\nu$.}

In this subsection, we give a sketch of the proof of the identity \eqref{eq:nutau}. We proceed as in \cite[Sect.4.2]{BFM}. The explicit form of the critical exponent 
follows from an exact computation of the `mass-mass' correlation function $S^{(2,2)}_{R,\o,\o'}(x,y)$, which goes as follows. We let 
\be G^{(4)}_{R,\o}(x,y,u,v):=Z^4\lim_{L,N,a^{-1}\to\infty}\media{\psi^-_{x,\o}\psi^-_{y,-\o}\psi^+_{u,-\o}\psi^+_{v,\o}}_{L,N,a}.\ee
Note that $G^{(4)}_{R,\o}$ is related to $S^{(2,2)}_{R,\o,\o'}$ via the identity 
\be G^{(4)}_{R,\o}(x,y,x,y)=\frac{Z^4}{(Z^{(2)})^2}S^{(2,2)}_{R,-\o,\o}(x,y).\label{G4S22}\ee
The key observation is that $G^{(4)}_\o$ and the dressed propagator $G^{(2)}_\o$ satisfy a system of two equations, obtained by combining the Ward Identities with the Schwinger-Dyson 
equation (see below for details):
\bea && D^x_\o G^{(2)}_{R,\o}(x,0)=Z\d(x)-\l_{\infty}F_{\o,-}(x)G^{(2)}_{R,\o}(x,0), \label{eq:C.25}\\
&& D^x_\o G^{(4)}_{R,\o}(x,y,u,v)=Z\d(x-v)G^{(2)}_{R,-\o}(y,u)+\label{eq:C.26}\\
&&\qquad +\l_\infty\big[F_{\o,+}(x-y)-F_{\o,+}(x-u)-F_{\o,-}(x-v)\big]G^{(4)}_{R,\o}(x,y,u,v), \nn\eea
where $D^x_\o$ was defined in \eqref{Dxo}, and $F_{\o,\e}(x)=\int_{\mathbb R^2}\frac{dp}{(2\p)^2}e^{ipx}\hat F_{\o,\e}(p)$, with 
\be \hat F_{\o,\e}(p):=\frac{\hat v_0(p)A_\e(p)}{D_{-\o}(p)},\qquad A_\e(p):=\frac12\Big(\frac1{1-\t\hat v_0(p)}+\e\frac1{1+\t\hat v_0(p)}\Big).\ee
The solution of \eqref{eq:C.25}-\eqref{eq:C.26} decaying to zero at infinity, as one can easily check by substitution, is 
\bea && \hskip-.3truecm  G^{(2)}_{R,\o}(x,y)=Z e^{-\l_\infty\D_{-}(x-y,0)} g_{R,\o}(x-y),\label{eq:C.28}\\
&& \hskip-.3truecm G^{(4)}_{R,\o}(x,y,u,v)=e^{\l_\infty[\D_{+}(x-y,v-y)-\D_{+}(x-u,v-u)]} 
G^{(2)}_{R,\o}(x,v)G^{(2)}_{R,-\o}(y,u),\nn
\eea
where $$g_{R,\o}(x)=\int_{\mathbb R^2} \frac{dk}{(2\p)^2}\frac{e^{-i k(x-y)}}{(-i-\o)k_1+(-i+\o)k_2},$$ and 
\be \D_{\e}(x,z):=\int_{\mathbb R^2}\frac{dp}{(2\p)^2}\frac{e^{ipz}-e^{ipx}}{D_\o(p)}\hat F_{\o,\e}(p).\ee
By using the fact that, asymptotically for large $|x|$,
\be \D_\e(x,0)\underset{|x|\to\infty}{\simeq}-\frac{A_\e(0)}{4\p}\log|x|, \ee
we find from \eqref{G4S22} and \eqref{eq:C.28} that 
\be S^{(2,2)}_{R,-\o,\o}(x,y) \underset{|x|\to\infty}{\simeq} ({\rm const.})|x-y|^{-2(1-\t)/(1+\t)},\ee
where we recall that $\t=-\l_\infty/(8\p)$. This proves \eqref{eq:nutau}.

\bigskip

We are left with giving a sketch of the proof of \eqref{eq:C.25}-\eqref{eq:C.26}. Let us start with \eqref{eq:C.25}. The idea is to combine the Ward Identities 
\eqref{C.22} with the so-called Schwinger-Dyson equation for the propagator, which reads
\be \hat G_{R,\o}^{(2)}(k)=\frac{Z}{D_\o(k)}\Big[1+\frac{\l_\infty}{Z^{(1)}}\int_{\mathbb R^2}\frac{dp}{(2\p)^2}\hat v_0(p)\hat G^{(2,1)}_{R,-\o,\o}(k,p)\Big].\label{C.33}\ee
For an elementary proof of this identity, see \cite[Appendix A.1]{BFM1}.

On the other hand, the two Ward Identities \eqref{C.22} imply a second, independent, relation between $\hat G^{(2,1)}_{R,-\o,\o}(k,p)$ and the dressed propagator: 
in fact, by solving \eqref{C.22} for $\hat G^{(2,1)}_{R,-\o,\o}(k,p)$ we find
\be \frac{Z}{Z^{(1)}}\hat D_{-\o}(p)\hat G^{(2,1)}_{R,-\o,\o}(k,p)=A_-(p)\big[\hat G^{(2)}_{R,\o}(k)-\hat G^{(2)}_{R,\o}(k+p)\big].\label{c.34}\ee
By plugging this equation into \eqref{C.33}, taking Fourier transform, and using the fact that $\int \hat F_{\o,\e}(p)\,dp=0$, we get \eqref{eq:C.25}, as desired. 

In order to prove \eqref{eq:C.26} we proceed analogously. The required analogue of the Ward Identity \eqref{eq:C.14} is obtained by deriving \eqref{eq:C.12}
w.r.t. $\phi^-_{z,\o}$, $\phi^-_{v,-\o}$, $\phi^+_{u,-\o}$ and $\phi^+_{y,\o}$, then setting the external sources to zero, and finally taking the removed cut-offs limit. The result is
\bea &&\big[ \d_{\o,\bar\o}\big(\d(x-z)-\d({x-y})\big)+\d_{\o,-\bar\o}\big(\d(x-v)-\d(x-u)\big)\big]G^{(4)}_\o(y,u,v,z)\nn\\
&&-\frac{Z}{Z^{(1)}}D_{\bar\o}^xG^{(4,1)}_{R,\bar\o,\o}(x,y,u,v,z)-\frac{Z}{Z^{(1)}}I_{\bar\o,\o}(x,y,u,v,z)=0,\label{Cc.13}
\eea
where
\bea && \hskip-.3truecm G^{(4,1)}_{R,\o',\o}(x,y,u,v,z):=Z^{(1)}Z^4\lim_{L,N,a^{-1}\to\infty}\media{\r^{(1)}_{x,\o'};\psi^-_{y,\o}\psi^-_{u,-\o}\psi^+_{v,-\o}\psi^+_{z,\o}}_{L,N,a},\nn\\
&& \hskip-.3truecm I_{\bar\o,\o}(x,y,u,v,z):=Z^{(1)}Z^4\lim_{L,N,a^{-1}\to\infty}\media{\d T_{\bar\o}(x);\psi^-_{y,\o}\psi^-_{u,-\o}\psi^+_{v,-\o}\psi^+_{z,\o}}_{L,N,a}.\nn\eea
We now take Fourier transform, with the conventions
\bea
&&G^{(4,1)}_{R,\o',\o}(x,y,u,v,z)=\int_{\mathbb R^2}\frac{dp}{(2\pi)^2} \int_{\mathbb R^2} \frac{dk_1}{(2\pi)^2}\int_{\mathbb R^2} \frac{dk_2}{(2\pi)^2}\int_{\mathbb R^2} \frac{dq}{(2\pi)^2}\times\nn\\
&&\qquad \qquad\times e^{ipx -i(k_1+p)y
-ik_2u+i(k_2-q)v+i(k_1+q)z} \hat G^{(4,1)}_{R,\o',\o}(p,k_1,k_2,q)\;,\nn\\
&&G^{(4)}_{R,\o}(y,u,v,z) = \int_{\mathbb R^2} \frac{dk_1}{(2\pi)^2}\int_{\mathbb R^2} \frac{dk_2}{(2\pi)^2}\int_{\mathbb R^2} \frac{dq}{(2\pi)^2}\times\\
&&\qquad\qquad  \times
e^{-ik_1y-ik_2u+i(k_2-q)v+i(k_1+q)z}\hat G^{(4)}_{R,\o}(k_1,k_2,q)\;,\nn\eea
and similarly for $I_{\o',\o}$. If we combine the two equations obtained from \eqref{Cc.13} by setting $\bar\o=\o$ and $\bar\o=-\o$,
and we use the fact that the anomaly term $\hat I_{\o',\o}$ has an expression completely analogous to \eqref{c.21}, namely
\be \hat I_{\o',\o}(p,k_1,k_2,q)=-\frac{\l_\infty}{8\p}D_{-\o'}(p)\hat v_0(p)\hat G^{(4,1)}_{-\o',\o}(p,k_1,k_2,q),\ee
we get
\bea && \frac{Z}{Z^{(1)}}D_{-\o}(p)\hat G^{(4,1)}_{R,-\o,\o}(p,k_1,k_2,q)\label{c.40}\\
&&\qquad =\, A_-(p)\big[\hat G^{(4)}_{R,\o}(k_1,k_2,q)-\hat G^{(4)}_{R,\o}(k_1+p,k_2,q)\big]\nn\\
&&\qquad +\ A_+(p)\big[\hat G^{(4)}_{R,\o}(k_1+p,k_2-p,q-p)-\hat G^{(4)}_{R,\o}(k_1+p,k_2,q-p)\big],\nn\eea
which is the analogue of \eqref{c.34}.

In addition to this Ward Identity, we need the analogue of \eqref{C.33}, that is the Schwinger-Dyson equation for the four-point function, which reads
\bea &&  \hat G_{R,\o}^{(4)}(k_1,k_2,q)=\frac{Z}{D_\o(k_1)}\Big[(2\p)^2\d(q)\hat G^{(2)}_{R,-\o}(k_2)\nn\\
&&\qquad\qquad +\frac{\l_\infty}{Z^{(1)}}\int_{\mathbb R^2}\frac{dp}{(2\p)^2}\hat v_0(p)\hat G^{(4,1)}_{R,-\o,\o}(p,k_1,k_2,q)\Big].\label{eq:C.41}\eea
If we now plug \eqref{c.40} into \eqref{eq:C.41} we finally obtain \eqref{eq:C.26}, as desired.

\bigskip

{\bf Acknowledgements}. We thank K. Gawedzki for several enlightening discussions on the connections between 
the two-dimensional massless Gaussian field and SLE processes. The main part of this work was completed during 
a sabbatical visit of A.G. to the Mathematics Department of the University Lyon 1, supported by the
A*MIDEX project Hypathie (n. ANR-11-IDEX-0001-02) funded by the ``Investissements d'Avenir''
25 French Government program, managed by the French National Research Agency (ANR), and by a CNRS visiting research grant, which are gratefully acknowledged. F.T. was partially supported by the CNRS PICS grant ``Interfaces al\'eatoires discs\`etes et dynamiques de Glauber''.

\end{document}